\documentclass[a4paper]{article}
\usepackage[T1]{fontenc}
\usepackage[utf8]{inputenc}
\usepackage{lmodern}
\usepackage{amsmath,amssymb,mathtools,amsthm}
\usepackage{graphicx}
\usepackage{geometry}
\usepackage{subcaption}
\usepackage{enumitem}
\usepackage{multirow}
\usepackage{booktabs}
\usepackage{comment}
\usepackage{graphicx, tikz}
\usetikzlibrary{matrix,calc}
\usetikzlibrary{patterns}
\usetikzlibrary{patterns.meta}
\usetikzlibrary{fadings,patterns}
\usetikzlibrary{backgrounds}
\usepackage{pdfrender} 
\newcommand{\backgroundcontour}[1]{%
\textpdfrender{
	TextRenderingMode=FillStrokeClip,
	LineWidth=1.75pt,
	FillColor=white,
	StrokeColor=white,
	MiterLimit=1
}{#1}%
}
\tikzset{
		invisible/.style = {minimum size=0pt, inner sep=0pt}
}
\newcommand{\stripedunderline}{%
\tikz\node[inner sep=2pt,minimum height=.75ex,minimum width=2ex,rectangle,overlay, shift={(-.5ex,-1.25ex)},pattern={Lines[angle=45,distance=2pt,line width=1pt]},pattern color=black!30]{};%
}
\newcommand{\solidunderline}{%
\tikz\node[inner sep=2pt,minimum height=.75ex,minimum width=2ex,rectangle,overlay, shift={(-.5ex,-1.25ex)},fill=black!30]{};%
}

\usepackage{xcolor}
\definecolor{colorA}{HTML}{B31B1B}
\definecolor{colorAdark}{HTML}{9B1717}
\definecolor{colorB}{HTML}{1BB3B3}
\definecolor{colorBdark}{HTML}{0F6161}
\definecolor{colorC}{HTML}{B38D1B}
\definecolor{colorCdark}{HTML}{6D5612}
\definecolor{colorD}{HTML}{1B3AB3}
\definecolor{colorDdark}{HTML}{152C8A}
\definecolor{colorE}{HTML}{AFB31B}
\definecolor{colorEdark}{HTML}{575A0C}
\definecolor{colorF}{HTML}{B31B80}
\definecolor{colorFdark}{HTML}{7C1359}

\usepackage[hyphens]{url}
\usepackage{hyperref}
\hypersetup{
	colorlinks=true,
	urlcolor = colorDdark,
	linkcolor = colorAdark,
	citecolor = colorCdark,
}
\usepackage[hyphenbreaks]{breakurl}
\usepackage{xurl}
\hypersetup{breaklinks=true}

\usepackage{cleveref}

\usepackage[ruled,vlined,linesnumbered]{algorithm2e}
\SetKw{KwDownTo}{downto}
\SetKw{KwBreak}{break}
\crefname{algocf}{alg.}{algs.}
\Crefname{algocf}{Algorithm}{Algorithms}

\SetKwComment{tcc}{$\triangleright$~}{}
\SetKwComment{tcp}{$\triangleright$~}{}
\SetCommentSty{normalfont}

\newtheorem{theorem}{Theorem}
\crefname{theorem}{theorem}{theorems}
\Crefname{theorem}{Theorem}{Theorems}

\newtheorem{corollary}{Corollary}
\crefname{corollary}{corollary}{corollaries}
\Crefname{corollary}{Corollary}{Corollaries}

\crefname{assumption}{assumption}{assumptions}
\Crefname{assumption}{Assumption}{Assumptions}

\theoremstyle{definition}
\newtheorem{lemma}{Lemma}
\crefname{lemma}{lemma}{lemmas}
\Crefname{lemma}{Lemma}{Lemmas}

\newtheorem{definition}{Definition}
\crefname{definition}{definition}{definitions}
\Crefname{definition}{Definition}{Definitions}

\newtheorem{observation}{Observation}
\crefname{observation}{observation}{observations}
\Crefname{observation}{Observation}{Observations}

\newtheorem{problem}{Problem}
\crefname{problem}{problem}{problems}
\Crefname{problem}{Problem}{Problems}

\theoremstyle{remark}

\newcommand{\qbeg}{q_{\text{s}}}
\newcommand{\qend}{q_{\text{e}}}
\newcommand{\tbeg}{t_{\text{s}}}
\newcommand{\tend}{t_{\text{e}}}
\newcommand{\wprec}{\prec_\mathrm{w}}
\newcommand{\chainxprec}{\prec_{\mathtt{ChainX}}}

\newcommand{\startanchor}{a_\mathrm{start}}
\newcommand{\finalanchor}{a_\mathrm{end}}
\DeclareMathOperator{\gscost}{\mathsf{gcost}}

\title{Practical colinear chaining on sequences revisited}
\usepackage{orcidlink}
\usepackage{authblk}
\author[1]{Nicola Rizzo}
\author[2]{Manuel C{\'a}ceres}
\author[3]{Veli M{\"a}kinen}
\affil[1]{Department of Computer Science, University of Helsinki, Finland, \texttt{nicola.rizzo@helsinki.fi} \orcidlink{0000-0002-2035-6309}}
\affil[2]{Department of Computer Science, Aalto University, Finland, \texttt{manuel.caceres@aalto.fi} \orcidlink{0000-0003-0235-6951}}
\affil[3]{Department of Computer Science, University of Helsinki, Finland, \texttt{veli.makinen@helsinki.fi} \orcidlink{0000-0003-4454-1493}}
\date{}

\DeclareMathOperator{\connect}{connect}
\DeclareMathOperator{\maxgap}{g}
\DeclareMathOperator{\diffoverlap}{o}
\DeclareMathOperator{\diag}{diag}
\DeclareMathOperator{\bucket}{bucket}

\begin{document}

\maketitle
\begin{abstract}
Colinear chaining is a classical heuristic for sequence alignment and is widely used in modern practical aligners.
Jain et al. (J.\ Comput.\ Biol.\ 2022) proposed an $O(n \log^3 n)$ time algorithm to chain a set of $n$ anchors so that the chaining cost matches the edit distance of the input sequences, when anchors are all the maximal exact matches. Moreover, assuming a uniform and sparse distribution of anchors, they provided a practical solution (\texttt{\upshape ChainX}) working in $O(n \cdot \mathrm{SOL} + n \log n)$ average-case time, where $\mathrm{SOL}$ is the cost of the output chain. 
This practical solution is not guaranteed to be optimal: we study the failing cases, introduce the \emph{anchor diagonal distance}, and find and implement an optimal algorithm working in $O(n \cdot \mathrm{OPT} + n \log n)$ average-case time, where $\mathrm{OPT}$ $\le \mathrm{SOL}$ is the optimal chaining cost. 
We validate the results by Jain et al., show that \texttt{\upshape ChainX} can be suboptimal with a realistic long read dataset, and show minimal computational slowdown for our solution.
\end{abstract}


\tableofcontents

\section{Introduction}

Colinear chaining is a popular technique to approximate the alignment of two sequences \cite{MM95,SK03,AO05,MS20,JainGT22}. It is used as one critical step in practical aligners like \texttt{minimap2} \cite{Li18} and \texttt{nucmer4} \cite{Marcais2018-hr}, and also speeds up exact alignment tools like \texttt{A*PA2}~\cite{koerkamp2022exact}. The main idea of the approach is to first identify short common parts between two sequences (seeds) as alignment anchors, and then to select a subset of these anchors that forms a linear ordering of seeds simultaneously on both sequences (colinear chain). A colinear chain can be converted into an alignment, but the approach is heuristic unless anchors, overlaps, and gaps are treated properly: namely, M\"akinen and Sahlin \cite{MS20}, while revisiting the work of Shibuya and Kurochkin \cite{SK03}, proposed a chaining formulation solvable in $O(n \log n)$ time (using sparse dynamic programming) considering both anchor overlaps and gaps between $n$ anchors so that the chaining score equals the longest common subsequence (LCS) length when all maximal exact matches (MEMs) are used as anchors. This connection requires that the MEMs are 
not filtered by length, a common optimization used by practical tools:
there are a quadratic number of MEM anchors on most pairs of sequences, so this approach does not directly yield an improvement over the classical dynamic programming approach \cite{Kuc19} nor violates the conditional lower bound for LCS~\cite{ABVW15}.
However, when a subset of MEMs is selected, for example by imposing uniqueness (see e.g.~\cite{Marcais2018-hr}) or a minimum length threshold, or even when the input is an arbitrary set of exact matches, the chaining formulation yields a non-trivial connection to a so-called \emph{anchored} version of LCS \cite{MS20}.

Recently, Jain et al.\ \cite{JainGT22} extended the colinear chaining framework by considering both overlap and gap costs, obtained analogous results to those of M\"akinen and Sahlin, and connected chaining to the classical Levenshtein distance (unit-cost edit distance), that---similarly to LCS---is also unlikely to be solved exactly in subquadratic time \cite{BK15}.
Jain et al.\ solved such chaining formulation in $O(n \log^{3} n)$ time. Moreover, the authors provided a practical solution working in $O(n \cdot \mathrm{SOL} + n \log n)$ average-case time, assuming a uniform and sparse distribution of anchors, where $\mathrm{SOL}$ is the cost of the output chain.
Such solution was implemented in tool \texttt{\upshape ChainX}~\cite{ChainX}.
We show that the \texttt{\upshape ChainX} solution, although verified experimentally, is not guaranteed to be optimal.

We start by introducing the problem in \Cref{sec:preliminaries}; we reformulate the \texttt{\upshape ChainX} algorithm and show the failing cases in \Cref{sec:pitfalls}; in \Cref{sec:diagonal} we introduce the \emph{anchor diagonal distance} and complete the strategy of \texttt{\upshape ChainX} to find a provably correct solution in the same $O(n \cdot \mathrm{OPT} + n \log n)$ average-case time, where $\mathrm{OPT} \le \mathrm{SOL}$ is the optimal chaining cost.
Finally, in \Cref{sec:experiments}, we verify the experimental results of Jain et al.~\cite{JainGT22}, perform tests on a realistic human dataset, and show that using maximal unique match (MUM) anchors our solution improves the chaining cost of approximately 2000 of 100k sampled long reads against a human reference.

\section{Preliminaries}\label{sec:preliminaries}
We denote integer interval $\lbrace x, x + 1, \dots, y \rbrace$ as $[x..y]$, or just $[x]$ if $x = y$.
Given a finite alphabet $\Sigma$ ($\Sigma = \lbrace \mathtt{A}, \mathtt{C}, \mathtt{G}, \mathtt{T} \rbrace$ for all examples of this paper), let $T$, $Q$ be two strings over $\Sigma$ of length $\lvert T \rvert$ and $\lvert Q \rvert$, respectively.
We indicate with $T[x..y]$ the string obtained by concatenating the characters of $T$ from the $x$-th to the $y$-th (strings are thus 1-indexed), and we call $T[x..y]$ a \emph{substring} of $T$.
We say that interval pair $a = ([\qbeg..\qend], [\tbeg..\tend])$ is an \emph{exact match anchor}, or just anchor, between $Q$ and $T$ if $Q[\qbeg..\qend] = T[\tbeg..\tend]$, with $\qbeg,\qend \in [1..\lvert Q \rvert]$, $\tbeg,\tend \in [1..\lvert T \rvert]$, $\qbeg \le \qend$, and $\tbeg \le \tend$.
Moreover, we use $a$ to denote a $k$-length anchor with $k = \qend - \qbeg + 1$.
See \Cref{fig:anchors}.

\begin{figure}[htp]
    \centering
\begin{tikzpicture}
	\tikzset{
		invisible/.style = {minimum size=0pt, inner sep=0pt}
	}

\matrix[
	matrix of nodes,
	minimum size=2mm,
	column sep={3.5mm,between origins},
	row sep={3mm,between origins},
	inner sep=0pt,
	nodes={anchor=base},
	anchor=west,
	font=\ttfamily,
] (Q) {
G&A&A&A&T&G&G&T&C&A&T&G&T&G&T&G&G&C&G&G&T&T&C&A&C&T\\
};

\matrix[
	matrix of nodes,
	minimum size=2mm,
	column sep={3.5mm,between origins},
	row sep={3mm,between origins},
	inner sep=0pt,
	nodes={anchor=base},
	anchor=west,
	font=\ttfamily,
] (T) at (0,-.75) {
G&A&A&A&T&T&T&G&T&C&C&G&T&C&A&T&G&T&T&C&C&G&G&T&C&A&C&T\\
};

\begin{scope}[on background layer]
\draw[-,very thick, colorA] ($(Q-1-1.north west) + (0,.05cm)$) -- node[invisible] (Qanchor0) {} ($(Q-1-5.north east) + (0,.05cm)$);
\draw[-,very thick, colorA] ($(T-1-1.north west) + (0,.05cm)$) -- node[invisible] (Tanchor0) {} ($(T-1-5.north east) + (0,.05cm)$);

\draw[-,very thick, colorB] ($(Q-1-7.south west) + (0,-.05cm)$) -- node[invisible] (Qanchor1) {} ($(Q-1-9.south east) + (0,-.05cm)$);
\draw[-,very thick, colorB] ($(T-1-8.south west) + (0,-.05cm)$) -- node[invisible] (Tanchor1) {} ($(T-1-10.south east) + (0,-.05cm)$);

\draw[-,very thick, colorC] ($(Q-1-7.north west) + (0,.05cm)$) -- node[invisible] (Qanchor2) {} ($(Q-1-12.north east) + (0,.05cm)$);
\draw[-,very thick, colorC] ($(T-1-12.north west) + (0,.05cm)$) -- node[invisible] (Tanchor2) {} ($(T-1-17.north east) + (0,.05cm)$);

\draw[-,very thick, colorD] ($(Q-1-18.south west) + (0,-.05cm)$) -- node[invisible] (Qanchor3) {} ($(Q-1-20.south east) + (0,-.05cm)$);
\draw[-,very thick, colorD] ($(T-1-21.south west) + (0,-.05cm)$) -- node[invisible] (Tanchor3) {} ($(T-1-23.south east) + (0,-.05cm)$);
\draw[-,very thick, colorE] ($(Q-1-21.north west) + (0,.05cm)$) -- node[invisible] (Qanchor4) {} ($(Q-1-23.north east) + (0,.05cm)$);
\draw[-,very thick, colorE] ($(T-1-18.north west) + (0,.05cm)$) -- node[invisible] (Tanchor4) {} ($(T-1-20.north east) + (0,.05cm)$);
\draw[-,very thick, colorF] ($(Q-1-22.south west) + (0,-.05cm)$) -- node[invisible] (Qanchor5) {} ($(Q-1-26.south east) + (0,-.05cm)$);
\draw[-,very thick, colorF] ($(T-1-24.south west) + (0,-.05cm)$) -- node[invisible] (Tanchor5) {} ($(T-1-28.south east) + (0,-.05cm)$);
\draw[-,colorA] (Qanchor0) to node[pos=0.6] (a1) {\backgroundcontour{$a_1$}} (Tanchor0);
\node[colorAdark] at (a1) {$a_1$};
\draw[-,colorB] (Qanchor1) to node[pos=0.4] (a2) {\backgroundcontour{$a_2$}} (Tanchor1);
\node[colorBdark] at (a2) {$a_2$};
\draw[-,colorC] (Qanchor2) to node[pos=0.6] (a3) {\backgroundcontour{$a_3$}} (Tanchor2);
\node[colorCdark] at (a3) {$a_3$};
\draw[-,colorD] (Qanchor3) to node[right=-5pt,pos=0.6] (a4) {\backgroundcontour{$a_4$}} (Tanchor3);
\node[colorDdark] at (a4) {$a_4$};
\draw[-,colorE] (Qanchor4) to node[pos=0.5] (a5) {\backgroundcontour{$a_5$}} (Tanchor4);
\node[colorEdark] at (a5) {$a_5$};
\draw[-,colorF] (Qanchor5) to node[pos=0.4] (a6) {\backgroundcontour{$a_6$}} (Tanchor5);
\node[colorFdark] at (a6) {$a_6$};
\end{scope}

	\node[left] at (Q.west) {$Q=$};
	\node[left] at (T.west) {$T=$};

	\node[colorAdark,anchor=base] at ($ (Q-1-1)+(0,-2.5ex) $) {$\qbeg$};
	\node[colorAdark,anchor=base] at ($ (Q-1-5)+(0,-2.5ex) $) {$\qend$};
	\node[colorAdark,anchor=base] at ($ (T-1-1)+(0,-2.5ex) $) {$\tbeg$};
	\node[colorAdark,anchor=base] at ($ (T-1-5)+(0,-2.5ex) $) {$\tend$};

\foreach \i in {1,5,10,15,20,25}
{
	\node[anchor=base] at ($ (Q-1-\i.north) + (0,2ex) $) {\scriptsize \i};
}
\end{tikzpicture}
    \caption{A set $\mathcal{A} = \lbrace a_1, \dots, a_6 \rbrace$ of exact match anchors between strings $T$ and $Q$.}\label{fig:anchors}
\end{figure}
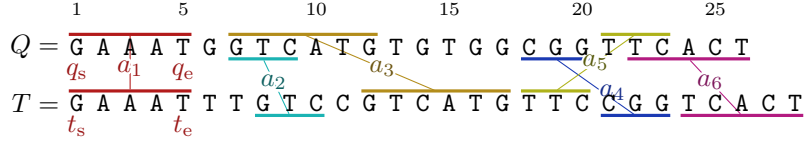

\begin{observation}[Exact match invariant]
    \label{obs:matchinvariant}
    The main invariant of an exact match anchor $a = ([\qbeg..\qend], [\tbeg..\tend])$ is that they indicate a substring of both $Q$ and $T$ of the same length, in symbols
$
    \qend - \qbeg + 1 = \tend - \tbeg + 1,
$
which can be rewritten as
$
    \qend - \qbeg = \tend - \tbeg
$
and also
$
    \qbeg - \tbeg = \qend - \tend.
$
\end{observation}

A colinear chaining formulation is defined by the concept of \emph{anchor precedence}---which sequences of anchors are coherent with an alignment and thus form a chain---and the concept of \emph{chain cost}, the score of a chain when the gaps and overlaps between the anchors are interpreted as insertions, deletions, and substitutions of 
the alignment these represent.

\begin{definition}[Anchor precedence~\cite{JainGT22}]\label{def:precedence}
    Let $a,a'$ be two exact match anchors between $Q,T \in \Sigma^+$ such that $a = ([\qbeg..\qend], [\tbeg..\tend])$ and $a' = ([\qbeg'..\qend'], [\tbeg'..\tend'])$.
    Then we say that:
\begin{itemize}
    \item $a$ \emph{strictly precedes} $a'$, in symbols $a \prec a'$, if
    $
        \qbeg \le \qbeg', \;
        \qend \le \qend', \;
        \tbeg \le \tbeg', \;
        \tend \le \tend',
    $
    and strict inequality holds for at least one of the four inequalities;
    \item $a$ \emph{weakly precedes} $a'$, in symbols $a_1 \wprec a_2$, if $\qbeg \le \qbeg'$, $\tbeg \le \tbeg'$, and strict inequality holds for at least one of the two inequalities;
    \item $a$ \emph{strongly precedes} $a'$, in symbols $a_1 \chainxprec a_2$, if
    $
        \qbeg < \qbeg', \;
        \qend < \qend', \;
        \tbeg < \tbeg'$, and $
        \tend < \tend'
    $
    (the notation $\chainxprec$ derives from its use in the \texttt{\upshape ChainX} implementation).
\end{itemize}
\end{definition}
Note that $a \chainxprec a'$ implies $a \wprec a'$ and $a \prec a'$.
In \Cref{fig:anchors}, we have that $a_2 \prec a_3$, $a_2 \wprec a_3$, but $a_2 \not \chainxprec a_3$.
We refer the reader to \cite{JainGT22} for insight on the difference between $\prec$ and $\wprec$ (which will not be needed in this paper).

\begin{problem}[Colinear chaining with overlap and gap costs~\cite{JainGT22}]\label{prob:CLC}
Given $T,Q \in \Sigma^+$, a sequence of anchors $a_1, \ldots, a_c$ between $Q$ and $T$ is called a \emph{colinear chain} (or just chain) if $a_i \prec a_{i+1}$ for all $i \in [1 .. c-1]$.
    The \emph{cost} of a chain $A = a_1, \ldots, a_c$ is $\gscost(A) = \sum_{i=0}^{c} \connect(a_i, a_{i+1})$, where as a convention $a_0 = \startanchor = ([0..0],[0..0])$, $a_{c+1} = \finalanchor = ([\lvert Q \rvert + 1..\lvert Q \rvert + 1], [\lvert T \rvert + 1..\lvert T \rvert + 1])$, and function $\connect$ on $a = ([\qbeg..\qend],[\tbeg..\tend])$, $a' = ([\qbeg'..\qend'],[\tbeg'..\tend'])$ is defined as
$
    \connect(a,a') = \maxgap(a,a') + \diffoverlap(a,a')
$
    with 
\begin{align*}
    \maxgap(a, a') &=
    \max\!\big(
        0,\;
        \qbeg' - \qend - 1,\;
        \tbeg' - \tend - 1
    \big),
    &\text{(gap cost)} \\
    \diffoverlap(a, a') &=
    \big\lvert
    \max\!\big(
        0,\;
        \qend - \qbeg' + 1
    \big) - \max\big(
        0,\;
        \tend - \tbeg' + 1
    \big) \big\rvert.
    &\text{(overlap cost)} 
\end{align*}
    Given set $\mathcal{A} = \lbrace a_1, \dots, a_n \rbrace$ of anchors between $Q$ and $T$, the \emph{colinear chaining problem with overlap and gap costs} consists of finding an ordered subset $A = a_1$, \dots, $a_c$ of $\mathcal{A}$ that is a colinear chain and has minimum cost.
\end{problem}

See \Cref{fig:chain} for an example of a colinear chain $A$ and the computation of $\gscost(A)$.
The colinear chaining problem can be also stated using weak precedence $\wprec$ and strong precedence $\chainxprec$: to discriminate between the three versions, we say that the formulation of \Cref{prob:CLC} is \emph{under strict precedence}, while the latter variants are \emph{under weak precedence} or \emph{under \texttt{\upshape ChainX} or \emph{strong} precedence}.

\begin{figure}[btp]
    \centering
    \begin{tikzpicture}
\matrix[
	matrix of nodes,
	minimum size=2mm,
	column sep={3.5mm,between origins},
	row sep={3mm,between origins},
	inner sep=0pt,
	nodes={anchor=base},
	anchor=west,
	font=\ttfamily,
] (Q) {
A&C&A&T&C&T&G&C&C&A&A&C&A&T&A&T&C&C\\
};

\matrix[
	matrix of nodes,
	minimum size=2mm,
	column sep={3.5mm,between origins},
	row sep={3mm,between origins},
	inner sep=0pt,
	nodes={anchor=base},
	anchor=west,
	font=\ttfamily,
] (T) at (0,-.75cm) {
A&C&A&T&C&C&G&G&C&C&A&T&A&T&A&T&C&C\\
};

\draw[-,very thick, colorA] ($(Q-1-2.north west) + (0,.05cm)$) -- node[invisible] (Qanchor0) {} ($(Q-1-5.north east) + (0,.05cm)$);
\draw[-,very thick, colorA] ($(T-1-2.north west) + (0,.05cm)$) -- node[invisible] (Tanchor0) {} ($(T-1-5.north east) + (0,.05cm)$);
\draw[-,very thick, colorB] ($(Q-1-7.south west) + (0,-.05cm)$) -- node[invisible] (Qanchor1) {} ($(Q-1-10.south east) + (0,-.05cm)$);
\draw[-,very thick, colorB] ($(T-1-8.south west) + (0,-.05cm)$) -- node[invisible] (Tanchor1) {} ($(T-1-11.south east) + (0,-.05cm)$);
\draw[-,very thick, colorC] ($(Q-1-12.north west) + (0,.05cm)$) -- node[invisible] (Qanchor2) {} ($(Q-1-16.north east) + (0,.05cm)$);
\draw[-,very thick, colorC] ($(T-1-10.north west) + (0,.05cm)$) -- node[invisible] (Tanchor2) {} ($(T-1-14.north east) + (0,.05cm)$);
\draw[-,very thick, colorD] ($(Q-1-13.south west) + (0,-.05cm)$) -- node[invisible] (Qanchor3) {} ($(Q-1-17.south east) + (0,-.05cm)$);
\draw[-,very thick, colorD] ($(T-1-13.south west) + (0,-.05cm)$) -- node[invisible] (Tanchor3) {} ($(T-1-17.south east) + (0,-.05cm)$);
\draw[-,colorA] (Qanchor0) to node[pos=0.7] (a1) {\backgroundcontour{$a_1$}} (Tanchor0);
\node[colorAdark] at (a1) {$a_1$};
\draw[-,colorB] (Qanchor1) to node[pos=0.3] (a2) {\backgroundcontour{$a_2$}} (Tanchor1);
\node[colorBdark] at (a2) {$a_2$};
\draw[-,colorC] (Qanchor2) to node[pos=0.7] (a3) {\backgroundcontour{$a_3$}} (Tanchor2);
\node[colorCdark] at (a3) {$a_3$};
\draw[-,colorD] (Qanchor3) to node[pos=0.3] (a4) {\backgroundcontour{$a_4$}} (Tanchor3);
\node[colorDdark] at (a4) {$a_4$};
\begin{scope}[on background layer]
\draw[draw=none,pattern={Lines[angle=45,distance=2pt,line width=1pt]},pattern color=black!30] ($(Q-1-6.north west)!0.5!(Q-1-5.north east)$) rectangle ($(Q-1-6.south east)!0.5!(Q-1-7.south west)$);
\draw[draw=none,pattern={Lines[angle=45,distance=2pt,line width=1pt]},pattern color=black!30] ($(T-1-6.north west)!0.5!(T-1-5.north east)$) rectangle ($(T-1-7.south east)!0.5!(T-1-8.south west)$);
\draw[draw=none,pattern={Lines[angle=45,distance=2pt,line width=1pt]},pattern color=black!30] ($(Q-1-11.north west)!0.5!(Q-1-10.north east)$) rectangle ($(Q-1-11.south east)!0.5!(Q-1-12.south west)$);
\draw[draw=none,fill=black!30] ($(T-1-10.north west)!0.5!(T-1-9.north east)$) rectangle ($ (T-1-12.south east)!0.5!(T-1-11.south west) $);
\draw[draw=none,fill=black!30] ($(Q-1-13.north west)!0.5!(Q-1-12.north east)$) rectangle ($ (Q-1-17.south east)!0.5!(Q-1-16.south west) $);
\draw[draw=none,fill=black!30] ($(T-1-13.north west)!0.5!(T-1-12.north east)$) rectangle ($ (T-1-15.south east)!0.5!(T-1-14.south west) $);
\draw[draw=none,pattern={Lines[angle=45,distance=2pt,line width=1pt]},pattern color=black!30] ($(Q-1-1.north west)$) rectangle ($(Q-1-1.south east)!0.5!(Q-1-2.south west)$);
\draw[draw=none,pattern={Lines[angle=45,distance=2pt,line width=1pt]},pattern color=black!30] ($(T-1-1.north west)$) rectangle ($(T-1-1.south east)!0.5!(T-1-2.south west)$);
\draw[draw=none,pattern={Lines[angle=45,distance=2pt,line width=1pt]},pattern color=black!30] ($(Q-1-17.north east)!0.5!(Q-1-18.north west)$) rectangle ($(Q-1-18.south east)$);
\draw[draw=none,pattern={Lines[angle=45,distance=2pt,line width=1pt]},pattern color=black!30] ($(T-1-17.north east)!0.5!(T-1-18.north west)$) rectangle ($(T-1-18.south east)$);
\end{scope}
    \node[anchor=west] at (0,-1.5cm) (cost) {$(\stripedunderline 1 + 0) + (\stripedunderline 2 + 0) + (\stripedunderline 1 + \solidunderline 2) + (0 + \solidunderline 2) + (\stripedunderline 1 + 0) = 9$};
	\node[left] at (Q.west) {$Q=$};
	\node[left] at (T.west) {$T=$};
	\node[left] at (cost.west) {$\gscost(A)=$};

\foreach \i in {2,5,7,10,12,13,16,17}
{
	\node[anchor=base] at ($ (Q-1-\i.north) + (0,2ex) $) {\scriptsize \i};
}
\foreach \i in {2,5,8,10,11,13,14,17}
{
	\node[anchor=base] at ($ (T-1-\i.south) - (0,2.5ex) $) {\scriptsize \i};
}
\end{tikzpicture}
    \caption{Example of a chain $A = a_1, a_2, a_3, a_4$ (under all precedence notions introduced) and of cost $\gscost(A) = 9$, where the gaps and overlaps between consecutive anchors are marked with striped and solid background, respectively.}\label{fig:chain}
\end{figure}

Next, we recall the main theoretical result from \cite{JainGT22}.
\begin{problem}[Anchored edit distance~\cite{JainGT22}]\label{prob:aED}
    Given $Q,T \in \Sigma^+$ and a set of anchors $\mathcal{A} = \lbrace a_1, \dots, a_n \rbrace$, we say that for any anchor $a = ([\qbeg..\qend],[\tbeg..\tend]) \in \mathcal{A}$ and integer $k \in [0..\tend - \tbeg]$ the character match $Q[\qbeg + k] = T[\tbeg + k]$ is \emph{supported}.
    Compute the optimal alignment between $Q$ and $T$ subject to the conditions that a match supported by some anchor has edit cost 0, a match that is not supported by any anchor has edit cost 1, and insertions, deletions, and substitutions have cost 1.
\end{problem}

\begin{theorem}[{\cite[Theorem 2]{JainGT22}}]\label{thm:edconnection}
    For a fixed set of anchors $\mathcal{A}$, the following quantities are equal: the anchored edit distance, the optimal colinear chaining cost under strict precedence, and the optimal colinear chaining cost under weak precedence.
\end{theorem}

\Cref{thm:edconnection} connects the chaining of (exact) matches to unit-cost edit distance: if the anchors in input are the set of all maximal exact matches (MEMs), then all character matches are supported and thus the cost of the optimal colinear chain is exactly equal to unit-cost edit distance.
Note however that \Cref{prob:CLC} and \Cref{thm:edconnection} admit an anchor set chosen arbitrarily.
An alternative interpretation of the result is that breaking each anchor into many length-1 anchors, obtaining a setting equivalent to that of anchored edit distance, does not change the cost of the optimal solution to chaining nor anchored edit distance.
This holds true even when we break the anchors into matches of different lenghts, for example, if we break down a match of length 5 into two matches of length 3 and 2.
Conversely, we can merge together anchors obtaining longer exact matches.
We formalize this intuition into the following corollary of \Cref{thm:edconnection}.

\begin{corollary}[Equivalence under perfect chain operations]\label{cor:maximality}
    For a given set of anchors $\mathcal{A} = \lbrace a_1, \dots, a_n \rbrace$, consider the new set $\mathcal{A}'$ obtained by breaking down each anchor $a_i = ([\qbeg..\qend], [\tbeg..\tend])$ into length-1 anchors $a_{i,k} = ([\qbeg + k],[\tbeg + k])$ for $i \in [1..n]$, $k \in [0..\qend - \qbeg]$.
    Conversely, let $a = ([\qbeg..\qend], [\tbeg..\tend])$ and $a' = ([\qbeg'..\qend'], [\tbeg'..\tend'])$ be two anchors of $\mathcal{A}$ that form a \emph{perfect chain}, that is, $a \prec a'$ and $\connect(a,a') = 0$: the overlap difference is $0$ and there is no gap in $Q$ nor in $T$. 
    Then we can merge $a$ and $a'$ into a longer anchor $a \cdot a' \coloneqq ([\qbeg..\qend'],[\tbeg..\tend'])$.
    Let $\mathcal{A}''$ be the set of anchors obtained from $\mathcal{A}$ by arbitrarily splitting anchors and merging those forming perfect chains.
    Then the anchored edit distance and the optimal colinear chaining cost of $\mathcal{A}$, $\mathcal{A}'$, and $\mathcal{A}''$ are the same value.
\end{corollary}
\begin{proof}
    The operation of breaking down any given anchor $a$ into smaller anchors does not change the character matches of $Q$ and $T$ supported by the anchors.
    Similarly, the inverse operation of merging two anchors forming a perfect chain, even when they overlap, does not change the matches supported by the anchors.
    Since the corresponding set of supported character matches is not affected by these operations, \Cref{thm:edconnection} guarantees that the optimal colinear chaining cost stays the same as well.
\end{proof}

\section{Pitfalls of practical colinear chaining}\label{sec:pitfalls}
Let $\mathcal{A} = \lbrace a_1, \dots, a_n \rbrace$ be a set of anchors between $Q$ and $T$.
Jain et al.~\cite{JainGT22} proposed a solution to \Cref{prob:CLC} on input $\mathcal{A}$ based on the following intuitive dynamic programming formulation.
Assume the anchors of $\mathcal{A}$ are already sorted by increasing starting position $\qbeg$ in $Q$.
Then we can compute $C[i]$ for $i \in [0..n+1]$, equal to the cost of the optimal chain ending with anchor $a_i$, using the following recursion:
\begin{align}
    C[0] &\,=\, 0 \notag\\
    C[i]
    &\,=\,
    \min_{j < i \;:\; a_j \prec a_i}
    C[j] + \connect(a_j, a_i) &\text{for $i \in [1..n]$.}\label{eq:recursion}
\end{align}
The formula is well-defined, since by construction $a_0 = \startanchor$ precedes all other anchors, and the correctness follows (for the most part) from the additivity of the $\connect$ function: the cost of adding anchor $a'$ to an optimal chain ending at $a$, with $a \prec a'$, depends only on the gap and overlap costs between $a$ and $a'$.

However, in \cite{JainGT22} it is not proven explicitly whether all optimal recursive cases under the strict (or even weak) precedence are considered by looking only back in the ordered list of anchors, or in other words, that the correctness in computing $C[i]$ is still valid when skipping some or all anchors $a_{j} \prec a_{i}$ with $j > i$ and such that the starting positions $\qbeg$ in $Q$ of $a_{i}$ and $a_{j}$ are equal.
Moreover, the actual implementation of \texttt{\upshape ChainX}~\cite{ChainX} uses the strong precedence formulation $\chainxprec$ (see \Cref{def:precedence}).
We leave the proof of equivalence of the optimal chaining cost under $\prec$ (or $\wprec$) and $\chainxprec$ to \Cref{sec:appendix} in the Appendix, and the main text we concentrate on \Cref{prob:CLC} under $\chainxprec$.

Jain et al.\ also develop a practical solution to colinear chaining~\cite[Algorithm 2]{JainGT22}. 
Intuitively, if we guess that the cost of the optimal chain is at most $B$, initially set to some fixed parameter $B_\text{start}$, then we do not have to consider all recursive cases from \Cref{eq:recursion} in computing $C[i]$.
In particular, we can skip all anchors $a_j \prec a_i$ such that $\connect(a_j,a_i) > B$, since they clearly cannot be used in a chain of cost at most $B$.
If $a_j = ([\qbeg^j..\qend^j],[\tbeg^j..\tend^j])$ and $a_i = ([\qbeg^i..\qend^i],[\tbeg^i..\tend^i])$, Jain et al.\ relax the previous constraint to condition
\begin{equation}
    \qbeg^i - \qbeg^j \le B \qquad
    (\text{with} \; j < i \;\text{and}\; a_j \chainxprec a_i), \label{eq:flawedcondition}
\end{equation}
and argue that if a chain of cost at most $B$ exists, then \Cref{eq:flawedcondition} holds for all adjacent anchors of an optimal-cost chain.
They implement this modified computation of values $C[i]$ as the main loop of \texttt{\upshape ChainX}.
After its execution, if $C[n+1] > B$, the algorithm updates guess $B$ to $B \cdot \alpha$, with $\alpha > 1$ a constant ramp-up factor.

Assuming that:
\begin{itemize}
    \item there are less anchors than $\lvert Q \rvert$ and $\lvert T \rvert$, that is, $n \le \min\big( \lvert Q \rvert, \lvert T \rvert \big)$; and
    \item the anchors are distributed uniformly, that is, the probability that anchor $a_i$'s interval in $Q$ starts at position $x$ (in symbols $\qbeg^i = x$) is equal to $1 / \lvert Q \rvert \le 1 / n$;
\end{itemize}
then \texttt{\upshape ChainX} yields a $O(\mathrm{SOL} \cdot n + n \log n)$-time solution, where $\mathrm{SOL}$ is the chaining cost of the output chain (and $B_\text{max} / \alpha < \mathrm{SOL} \le B_\text{max}$ with $B_\text{max}$ the last value of $B$ considered).

We note the following error in the main strategy of \texttt{\upshape ChainX}.
\begin{observation}
\label{obs:wrong}
    Plugging condition $\qbeg^i - \qbeg^j \le B$ (\Cref{eq:flawedcondition}) into the subscript of the $\min$ operator in \Cref{eq:recursion} does not guarantee that all chains of cost at most $B$ are considered, including optimal chains, and \texttt{\upshape ChainX}~\cite[Algorithm 2]{JainGT22} is not always correct.
    Indeed, consider $T,Q \in \Sigma^+$ with $\lvert T \rvert = \lvert Q \rvert = 13$ and anchor set $\mathcal{A} = a_1, a_2, a_3, a_4$ with $a_1 = ([1..7],[1..7])$, $a_2 = ([9..12],[7..10])$, $a_3 = ([7..10],[9..12])$, and $a_4 = ([11..13], [11..13])$ as visualized in \Cref{fig:failing}.
    Then $\connect(a_1, a_4) = 3$, distance $\qbeg^4 - \qbeg^1$ is equal to $10$, optimal chain $a_1, a_4$ has cost $3$, whereas
    suboptimal chains $a_1, a_2, a_4$ and $a_1, a_3, a_4$ have cost $4$.
    If $B_\mathrm{start} = 9$, recursive case $\connect(a_1, a_4)$ is skipped but the suboptimal chains are considered, and thus $C[n+1] = 4 \le 9 = B$ and the termination condition is met.\footnote{The implementation of \texttt{\upshape ChainX}, that uses value $B_\text{start} = 100$, returns a suboptimal chain on a similarly crafted example where the initial runs of $\mathtt{A}$ in $Q$ and $T$ have length 98 and maximal unique match (MUM) anchors of length at least 3 are computed.}
\end{observation}
\begin{figure}[htp]
\begin{subfigure}{.5\textwidth}\centering\begin{tikzpicture}
\matrix[
	matrix of nodes,
	minimum size=2mm,
	column sep={3.5mm,between origins},
	row sep={3mm,between origins},
	inner sep=0pt,
	nodes={anchor=base},
	anchor=west,
	font=\ttfamily,
] (Q) {
A&A&A&A&A&A&T&G&T&C&T&C&C\\
};

\matrix[
	matrix of nodes,
	minimum size=2mm,
	column sep={3.5mm,between origins},
	row sep={3mm,between origins},
	inner sep=0pt,
	nodes={anchor=base},
	anchor=west,
	font=\ttfamily,
] (T) at (0,-.75) {
A&A&A&A&A&A&T&C&T&G&T&C&C\\
};

\draw[-,very thick, colorA] ($(Q-1-1.north west) + (0,.05cm)$) -- node[invisible] (Qanchor0) {} ($(Q-1-7.north east) + (0,.05cm)$);
\draw[-,very thick, colorA] ($(T-1-1.north west) + (0,.05cm)$) -- node[invisible] (Tanchor0) {} ($(T-1-7.north east) + (0,.05cm)$);
\draw[-,colorA] (Qanchor0) to node[pos=0.6] (a1) {\backgroundcontour{$a_1$}} (Tanchor0);
\node[colorAdark] at (a1) {$a_1$};

\draw[-,very thick, colorB] ($(Q-1-9.south west) + (0,-.05cm)$) -- node[invisible] (Qanchor1) {} ($(Q-1-12.south east) + (0,-.05cm)$);
\draw[-,very thick, colorB] ($(T-1-7.south west) + (0,-.05cm)$) -- node[invisible] (Tanchor1) {} ($(T-1-10.south east) + (0,-.05cm)$);\draw[-,colorB] (Qanchor1) to node[pos=1] (a2) {\backgroundcontour{$a_2$}} (Tanchor1);
\node[colorBdark] at (a2) {$a_2$};

\draw[-,very thick, colorC] ($(Q-1-7.north west) + (0,.1cm)$) -- node[invisible] (Qanchor2) {} ($(Q-1-10.north east) + (0,.1cm)$);
\draw[-,very thick, colorC] ($(T-1-9.north west) + (0,.1cm)$) -- node[invisible] (Tanchor2) {} ($(T-1-12.north east) + (0,.1cm)$);\draw[-,colorC] (Qanchor2) to node[pos=0] (a3) {\backgroundcontour{$a_3$}} (Tanchor2);
\node[colorCdark] at (a3) {$a_3$};

\draw[-,very thick, colorD] ($(Q-1-11.south west) + (0,-.1cm)$) -- node[invisible] (Qanchor3) {} ($(Q-1-13.south east) + (0,-.1cm)$);
\draw[-,very thick, colorD] ($(T-1-11.south west) + (0,-.1cm)$) -- node[invisible] (Tanchor3) {} ($(T-1-13.south east) + (0,-.1cm)$);
\draw[-,colorD] (Qanchor3) to node[pos=0.25] (a4) {\backgroundcontour{$a_4$}} (Tanchor3);
\node[colorDdark] at (a4) {$a_4$};
\node[left] at (Q.west) {$Q=$};
\node[left] at (T.west) {$T=$};
\node at (0,-1.5) {$\vphantom{\gscost(A)}$};
\foreach \i in {1,7,9,10,11,12,13}
{
	\node[anchor=base] at ($ (Q-1-\i.north) + (0,2ex) $) {\scriptsize \i};
}
\foreach \i in {1,7,9,10,11,12,13}
{
	\node[anchor=base] at ($ (T-1-\i.south) - (0,2.5ex) $) {\scriptsize \i};
}
\end{tikzpicture}\caption{Anchors $\mathcal{A}$ in input.}\label{fig:failing:anchors}\end{subfigure}%
\begin{subfigure}{.5\textwidth}\centering%
\begin{tikzpicture}
\matrix[
	matrix of nodes,
	minimum size=2mm,
	column sep={3.5mm,between origins},
	row sep={3mm,between origins},
	inner sep=0pt,
	nodes={anchor=base},
	anchor=west,
	font=\ttfamily,
] (Q) {
A&A&A&A&A&A&T&G&T&C&T&C&C\\
};

\matrix[
	matrix of nodes,
	minimum size=2mm,
	column sep={3.5mm,between origins},
	row sep={3mm,between origins},
	inner sep=0pt,
	nodes={anchor=base},
	anchor=west,
	font=\ttfamily,
] (T) at (0,-.75) {
A&A&A&A&A&A&T&C&T&G&T&C&C\\
};

\draw[-,very thick, colorA] ($(Q-1-1.north west) + (0,.05cm)$) -- node[invisible] (Qanchor0) {} ($(Q-1-7.north east) + (0,.05cm)$);
\draw[-,very thick, colorA] ($(T-1-1.north west) + (0,.05cm)$) -- node[invisible] (Tanchor0) {} ($(T-1-7.north east) + (0,.05cm)$);
\draw[-,very thick, colorC] ($(Q-1-7.south west) + (0,-.05cm)$) -- node[invisible] (Qanchor1) {} ($(Q-1-10.south east) + (0,-.05cm)$);
\draw[-,very thick, colorC] ($(T-1-9.south west) + (0,-.05cm)$) -- node[invisible] (Tanchor1) {} ($(T-1-12.south east) + (0,-.05cm)$);
\draw[-,very thick, colorD] ($(Q-1-11.north west) + (0,.05cm)$) -- node[invisible] (Qanchor2) {} ($(Q-1-13.north east) + (0,.05cm)$);
\draw[-,very thick, colorD] ($(T-1-11.north west) + (0,.05cm)$) -- node[invisible] (Tanchor2) {} ($(T-1-13.north east) + (0,.05cm)$);
\draw[-,colorA] (Qanchor0) to node[pos=0.6] (a1) {\backgroundcontour{$a_1$}} (Tanchor0);
\node[colorAdark] at (a1) {$a_1$};
\draw[-,colorC] (Qanchor1) to node[pos=0.4] (a2) {\backgroundcontour{$a_3$}} (Tanchor1);
\node[colorCdark] at (a2) {$a_3$};
\draw[-,colorD] (Qanchor2) to node[pos=0.6] (a3) {\backgroundcontour{$a_4$}} (Tanchor2);
\node[colorDdark] at (a3) {$a_4$};
\begin{scope}[on background layer]
\draw[draw=none,fill=black!30] ($(Q-1-7.north west)!0.5!(Q-1-6.north east)$) rectangle ($ (Q-1-8.south east)!0.5!(Q-1-7.south west) $);
\draw[draw=none,pattern={Lines[angle=45,distance=2pt,line width=1pt]},pattern color=black!30] ($(T-1-8.north west)!0.5!(T-1-7.north east)$) rectangle ($(T-1-8.south east)!0.5!(T-1-9.south west)$);
\draw[draw=none,pattern={Lines[angle=45,distance=2pt,line width=1pt]},pattern color=black!30] ($(Q-1-11.north west)!0.5!(Q-1-10.north east)$) rectangle ($(Q-1-10.south east)!0.5!(Q-1-11.south west)$);
\draw[draw=none,fill=black!30] ($(T-1-11.north west)!0.5!(T-1-10.north east)$) rectangle ($ (T-1-13.south east)!0.5!(T-1-12.south west) $);
\end{scope}
	\node[anchor=west] at (0,-1.5) (cost) {$0 + (\stripedunderline 1 + \solidunderline 1) + (0 + \solidunderline 2) + 0 = 4$};
	\node[left] at (Q.west) {$Q=$};
	\node[left] at (T.west) {$T=$};
	\node[left] at (cost.west) {$\gscost(A)=$};

\foreach \i in {1,7,10,11,13}
{
	\node[anchor=base] at ($ (Q-1-\i.north) + (0,2ex) $) {\scriptsize \i};
}
\foreach \i in {1,7,9,11,12,13}
{
	\node[anchor=base] at ($ (T-1-\i.south) - (0,2.5ex) $) {\scriptsize \i};
}
\end{tikzpicture}\caption{Suboptimal chain.}\label{fig:failing:suboptimal1}
\end{subfigure}\\[.5cm]
\begin{subfigure}{.5\textwidth}\centering%
\begin{tikzpicture}
\matrix[
	matrix of nodes,
	minimum size=2mm,
	column sep={3.5mm,between origins},
	row sep={3mm,between origins},
	inner sep=0pt,
	nodes={anchor=base},
	anchor=west,
	font=\ttfamily,
] (Q) {
A&A&A&A&A&A&T&G&T&C&T&C&C\\
};

\matrix[
	matrix of nodes,
	minimum size=2mm,
	column sep={3.5mm,between origins},
	row sep={3mm,between origins},
	inner sep=0pt,
	nodes={anchor=base},
	anchor=west,
	font=\ttfamily,
] (T) at (0,-.75) {
A&A&A&A&A&A&T&C&T&G&T&C&C\\
};

\draw[-,very thick, colorA] ($(Q-1-1.north west) + (0,.05cm)$) -- node[invisible] (Qanchor0) {} ($(Q-1-7.north east) + (0,.05cm)$);
\draw[-,very thick, colorA] ($(T-1-1.north west) + (0,.05cm)$) -- node[invisible] (Tanchor0) {} ($(T-1-7.north east) + (0,.05cm)$);
\draw[-,very thick, colorB] ($(Q-1-9.south west) + (0,-.05cm)$) -- node[invisible] (Qanchor1) {} ($(Q-1-12.south east) + (0,-.05cm)$);
\draw[-,very thick, colorB] ($(T-1-7.south west) + (0,-.05cm)$) -- node[invisible] (Tanchor1) {} ($(T-1-10.south east) + (0,-.05cm)$);
\draw[-,very thick, colorD] ($(Q-1-11.north west) + (0,.05cm)$) -- node[invisible] (Qanchor2) {} ($(Q-1-13.north east) + (0,.05cm)$);
\draw[-,very thick, colorD] ($(T-1-11.north west) + (0,.05cm)$) -- node[invisible] (Tanchor2) {} ($(T-1-13.north east) + (0,.05cm)$);
\draw[-,colorA] (Qanchor0) to node[pos=0.6] (a1) {\backgroundcontour{$a_1$}} (Tanchor0);
\node[colorAdark] at (a1) {$a_1$};
\draw[-,colorB] (Qanchor1) to node[pos=0.4] (a2) {\backgroundcontour{$a_2$}} (Tanchor1);
\node[colorBdark] at (a2) {$a_2$};
\draw[-,colorD] (Qanchor2) to node[pos=0.6] (a3) {\backgroundcontour{$a_4$}} (Tanchor2);
\node[colorDdark] at (a3) {$a_4$};
\begin{scope}[on background layer]
\draw[draw=none,pattern={Lines[angle=45,distance=2pt,line width=1pt]},pattern color=black!30] ($(Q-1-8.north west)!0.5!(Q-1-7.north east)$) rectangle ($(Q-1-8.south east)!0.5!(Q-1-9.south west)$);
\draw[draw=none,fill=black!30] ($(T-1-7.north west)!0.5!(T-1-6.north east)$) rectangle ($ (T-1-8.south east)!0.5!(T-1-7.south west) $);
\draw[draw=none,fill=black!30] ($(Q-1-11.north west)!0.5!(Q-1-10.north east)$) rectangle ($ (Q-1-13.south east)!0.5!(Q-1-12.south west) $);
\draw[draw=none,pattern={Lines[angle=45,distance=2pt,line width=1pt]},pattern color=black!30] ($(T-1-11.north west)!0.5!(T-1-10.north east)$) rectangle ($(T-1-10.south east)!0.5!(T-1-11.south west)$);
\end{scope}
	\node[anchor=west] at (0,-1.5) (cost) {$0 + (\stripedunderline 1 + \solidunderline 1) + (0 + \solidunderline 2) + 0 = 4$};

	\node[left] at (Q.west) {$Q=$};
	\node[left] at (T.west) {$T=$};
	\node[left] at (cost.west) {$\gscost(A)=$};
\foreach \i in {1,7,9,11,12,13}
{
	\node[anchor=base] at ($ (Q-1-\i.north) + (0,2ex) $) {\scriptsize \i};
}
\foreach \i in {1,7,10,11,13}
{
	\node[anchor=base] at ($ (T-1-\i.south) - (0,2.5ex) $) {\scriptsize \i};
}
\end{tikzpicture}\caption{Suboptimal chain.}\label{fig:failing:suboptimal2}
\end{subfigure}
\begin{subfigure}{.5\textwidth}\centering%
\begin{tikzpicture}
\matrix[
	matrix of nodes,
	minimum size=2mm,
	column sep={3.5mm,between origins},
	row sep={3mm,between origins},
	inner sep=0pt,
	nodes={anchor=base},
	anchor=west,
	font=\ttfamily,
] (Q) {
A&A&A&A&A&A&T&G&T&C&T&C&C\\
};

\matrix[
	matrix of nodes,
	minimum size=2mm,
	column sep={3.5mm,between origins},
	row sep={3mm,between origins},
	inner sep=0pt,
	nodes={anchor=base},
	anchor=west,
	font=\ttfamily,
] (T) at (0,-.75) {
A&A&A&A&A&A&T&C&T&G&T&C&C\\
};

\draw[-,very thick, colorA] ($(Q-1-1.north west) + (0,.05cm)$) -- node[invisible] (Qanchor0) {} ($(Q-1-7.north east) + (0,.05cm)$);
\draw[-,very thick, colorA] ($(T-1-1.north west) + (0,.05cm)$) -- node[invisible] (Tanchor0) {} ($(T-1-7.north east) + (0,.05cm)$);
\draw[-,very thick, colorD] ($(Q-1-11.south west) + (0,-.05cm)$) -- node[invisible] (Qanchor1) {} ($(Q-1-13.south east) + (0,-.05cm)$);
\draw[-,very thick, colorD] ($(T-1-11.south west) + (0,-.05cm)$) -- node[invisible] (Tanchor1) {} ($(T-1-13.south east) + (0,-.05cm)$);
\draw[-,colorA] (Qanchor0) to node[pos=0.6] (a1) {\backgroundcontour{$a_1$}} (Tanchor0);
\node[colorAdark] at (a1) {$a_1$};
\draw[-,colorD] (Qanchor1) to node[pos=0.4] (a2) {\backgroundcontour{$a_4$}} (Tanchor1);
\node[colorDdark] at (a2) {$a_4$};
\begin{scope}[on background layer]
\draw[draw=none,pattern={Lines[angle=45,distance=2pt,line width=1pt]},pattern color=black!30] ($(Q-1-8.north west)!0.5!(Q-1-7.north east)$) rectangle ($(Q-1-10.south east)!0.5!(Q-1-11.south west)$);
\draw[draw=none,pattern={Lines[angle=45,distance=2pt,line width=1pt]},pattern color=black!30] ($(T-1-8.north west)!0.5!(T-1-7.north east)$) rectangle ($(T-1-10.south east)!0.5!(T-1-11.south west)$);
\end{scope}
	\node[anchor=west] at (0,-1.5) (cost) {$0 + (\stripedunderline 3 + 0) + 0 = 3$};

	\node[left] at (Q.west) {$Q=$};
	\node[left] at (T.west) {$T=$};
	\node[left] at (cost.west) {$\gscost(A)=$};
\foreach \i in {1,7,11,13}
{
	\node[anchor=base] at ($ (Q-1-\i.north) + (0,2ex) $) {\scriptsize \i};
}
\foreach \i in {1,7,11,13}
{
	\node[anchor=base] at ($ (T-1-\i.south) - (0,2.5ex) $) {\scriptsize \i};
}
\end{tikzpicture}\caption{Optimal chain.}\label{fig:failing:optimal}
\end{subfigure}
\caption{Example where \texttt{\upshape ChainX} with $B_\mathrm{start} = 9$ returns a suboptimal chain.}\label{fig:failing}
\end{figure}

Note that substituting the condition of \Cref{eq:flawedcondition} with condition
$
    \qbeg^i - \qend^j - 1 \le B
$
considers correctly all anchors $a_j$ with gap cost in $Q$ of at most $B$ but is harder to compute efficiently, since $\mathcal{A}$ is sorted by $\qbeg$ and all anchor overlaps in $Q$ are also included: additionally sorting $\mathcal{A}$ by end position in $Q$ (i.e.\ $\qend$) would make it simpler to consider all anchor overlaps, but the total number of comparisons\footnote{While we do assume a uniform distribution of anchors, we do not assume anything about their length. Indeed, long MUM and MEM anchors are expected in the comparison of similar strings, intuitively resulting in many anchor overlaps.} could be $O(n^2)$.
On the other hand, when anchors of $\mathcal{A}$ are exact match anchors of a fixed constant length $k \in O(1)$ like in the case of minimizer anchors, using the simple condition
$
\qbeg^i - \qbeg^j \le B + k
$ does obtain optimality.
In the next section, we show complete the strategy of \texttt{\upshape ChainX} to obtain an optimal chain regardless of the anchors in input.

\section{Chaining with the diagonal distance}\label{sec:diagonal}
In the previous section, we showed that \texttt{\upshape ChainX} approximates a correct strategy---guessing an optimal chaining cost of $B$ and considering only anchor pairs $(a_j,a_i)$ of distance at most $B$ in $Q$---but sacrifices correctness (\Cref{obs:wrong}).
Consider how function $\connect$ (\Cref{prob:CLC}) is defined: 
given $a_j \prec a_i$, if $a_j$ and $a_i$ do not overlap in $Q$, in symbols $\qbeg^i - \qend^j > 0$, then we can replace the condition from \Cref{eq:flawedcondition} with $\qbeg^i - \qend^j - 1 \le B$.
On the other hand, if $a_j$ and $a_i$ do overlap in $Q$, then we do not want to consider all of these cases explicitly, as described at the end of the last section.
Keeping this in mind, we tentatively rewrite \Cref{eq:recursion} as
\begin{align}
    C[i] &=
    \min \big(
        \mathrm{gap}_Q(i),
        \mathrm{overlap}_Q(i)
        \big), \qquad \text{with} \label{eq:recursion2}\\
    \mathrm{gap}_Q(i) &=
    \min \Big\lbrace C[j] + \connect(a_j, a_i) \;\Big|\; j < i : \substack{a_j \prec a_i,\\
    0 \le \qbeg^i - \qend^j - 1 \le B,\\
    \connect(a_j,a_i) \le B} \Big\rbrace,\notag\\
    \mathrm{overlap}_Q(i) &=
    \min \Big\lbrace C[j] + \connect(a_j, a_i) \;\Big|\; j < i : \substack{a_j \prec a_i,\\ \qbeg^j \le \qbeg^i \le \qend^j,\\ \connect(a_j,a_i) \le B} \Big\rbrace, \notag
\end{align}
where $\mathrm{gap}_Q(i)$ considers all anchors $a_j$ with no overlap in $Q$, and $\mathrm{overlap}_Q(i)$ considers those with an overlap in $Q$.
In the example of \Cref{fig:anchors}, if $B = 11$ then $\mathrm{gap}_Q(6)$ considers $a_4$ and $a_3$ and $\mathrm{overlap}_Q(6)$ considers $a_5$.
To compute $\mathrm{gap}_Q(i)$ we use the original strategy of \texttt{\upshape ChainX}, whereas we devise an efficient way of computing $\mathrm{overlap}_Q(i)$ based on \emph{diagonal distance}.
This simple metric has been used in some chaining formulations (see \cite{sahlin2023survey}) but has not yet been connected to \Cref{prob:CLC}, to the best of our knowledge.
See also \Cref{fig:diagonal}.

\begin{lemma}\label{lem:diagonalcost}
    Given $Q,T \in \Sigma^+$, if $a = ([\qbeg..\qend],[\tbeg..\tend])$ is an exact match anchor between $Q$ and $T$, we define the \emph{diagonal} of $a$ as $\diag(a) = \qbeg - \tbeg$.
    Let $a = ([\qbeg..\qend],[\tbeg..\tend])$ and  $a' = ([\qbeg'..\qend'],[\tbeg'..\tend'])$ be anchors between $Q$ and $T$ such that $a \prec a'$ or $a \wprec a'$.
    Then:
\begin{itemize}[nosep]
    \item (gap-gap case) if $a$ and $a'$ do not overlap in $Q$ nor in $T$, in symbols $\qend < \qbeg'$ and $\tend < \tbeg'$, then $\connect(a, a') = \maxgap(a,a') = \max(\qbeg' - \qend - 1, \,\tbeg' - \tend - 1)$;
    \item in all other cases, $\connect(a, a') = \lvert \diag(a) - \diag(a') \rvert$, a value that we call \emph{diagonal distance}.
\end{itemize}
\end{lemma}

\begin{figure}[ht]
\centering
\begin{subfigure}{.5\textwidth}
\centering\scalebox{0.73}{\begin{tikzpicture}
	\tikzset{
		invisible/.style = {minimum size=0pt, inner sep=0pt}
	}

\matrix[
	matrix of nodes,
	minimum size=2mm,
	column sep={3.5mm,between origins},
	row sep={3.5mm,between origins},
	inner sep=1pt,
	nodes={anchor=base},
	anchor=south west,
	font=\ttfamily,
    ampersand replacement=\&
] (Q) at (0,0) {
A\&C\&A\&T\&C\&T\&G\&C\&C\&A\&A\&C\&A\&T\&A\&T\\
};

\matrix[
	matrix of nodes,
	minimum size=2mm,
	column sep={3.5mm,between origins},
	row sep={3.5mm,between origins},
	inner sep=1pt,
	nodes={anchor=base},
	anchor=north east,
	font=\ttfamily,
] (T) at (0,0) {
T\\C\\A\\T\\C\\C\\G\\G\\C\\C\\A\\T\\A\\C\\A\\A\\
};

\begin{scope}[on background layer]
\draw[black!15,step=3.5mm] (0,0) grid ($(Q-1-16.south east|-T-16-1.south east) + (1.1pt,-.01pt)$);
\end{scope}

\draw[-,very thick, colorA] ($(Q-1-2.north west) + (0,.05cm)$) -- node[invisible] (Qanchor0) {} ($(Q-1-5.north east) + (0,.05cm)$);
\draw[-,very thick, colorA] ($(T-2-1.north west) + (-.05cm,0)$) -- node[invisible] (Tanchor0) {} ($(T-5-1.south west) + (-.05cm,0)$);

\node[invisible] at ($(Q-1-2.west|-T-2-1.north) + (-1pt,1pt)$) (a0northwest) {};
\node[invisible] at ($(Q-1-5.west|-T-5-1.north) + (3.5mm,-3.5mm) + (-1pt,1pt)$) (a0southeast) {};
\draw[-,very thick,draw=colorA] (a0northwest) rectangle (a0southeast) node[invisible,pos=.5] (a0) {};

\draw[dashed,thick,-,colorA] (a0northwest) -- ($(Q-1-2.south west-|a0northwest)$);
\draw[dashed,thick,-,colorA] (a0southeast) -- ($(Q-1-5.south east-|a0southeast)$);
\draw[dashed,thick,-,colorA] (a0northwest) -- ($(T-2-1.north east|-a0northwest)$);
\draw[dashed,thick,-,colorA] (a0southeast) -- ($(T-5-1.south east|-a0southeast)$);

\draw[-,very thick, colorB] ($(Q-1-12.north west) + (0,.10cm)$) -- node[invisible] (Qanchor1) {} ($(Q-1-15.north east) + (0,.10cm)$);
\draw[-,very thick, colorB] ($(T-10-1.north west) + (-.10cm,0)$) -- node[invisible] (Tanchor1) {} ($(T-13-1.south west) + (-.10cm,0)$);

\node[invisible] at ($(Q-1-12.west|-T-10-1.north) + (-1pt,1pt)$) (a1northwest) {};
\node[invisible] at ($(Q-1-15.west|-T-13-1.north) + (3.5mm,-3.5mm) + (-1pt,1pt)$) (a1southeast) {};
\draw[-,very thick,draw=colorB] (a1northwest) rectangle (a1southeast) node[invisible,pos=.5] (a1) {};

\draw[dashed,thick,-,colorB] (a1northwest) -- ($(Q-1-12.south west-|a1northwest)$);
\draw[dashed,thick,-,colorB] (a1southeast) -- ($(Q-1-15.south east-|a1southeast)$);
\draw[dashed,thick,-,colorB] (a1northwest) -- ($(T-10-1.north east|-a1northwest)$);
\draw[dashed,thick,-,colorB] (a1southeast) -- ($(T-13-1.south east|-a1southeast)$);

\node[anchor=base,colorAdark] at (a0) {\backgroundcontour{$a$}};
\node[anchor=base,colorAdark] at (a0) {$a$};

\node[anchor=base,colorBdark] at (a1) {\backgroundcontour{$a'$}};
\node[anchor=base,colorBdark] at (a1) {$a'$};

\begin{scope}[on background layer]
\draw[draw=none,pattern={Lines[angle=45,distance=2pt,line width=1pt]},pattern color=black!30] ($(a0southeast)$) rectangle ($(a1northwest)$) node[pos=.5] (cost0) {};
\end{scope}
\node[align=center] at (cost0) {\backgroundcontour{$\connect(a,a')$} \\ \backgroundcontour{$=$} \\ \backgroundcontour{$6 + 0$}};
\node[align=center] at (cost0) {$\connect(a,a')$ \\ $=$ \\ $6 + 0$};

\begin{scope}[on background layer]
\draw[draw=none,pattern={Lines[angle=45,distance=2pt,line width=1pt]},pattern color=black!30] ($(Q-1-6.north west)!0.5!(Q-1-5.north east)$) rectangle ($(Q-1-11.south east)!0.5!(Q-1-12.south west)$);
\draw[draw=none,pattern={Lines[angle=45,distance=2pt,line width=1pt]},pattern color=black!30] ($(T-6-1.north west)$) rectangle ($(T-9-1.south east)$);
\end{scope}

\begin{scope}[on background layer]
\end{scope}
\node[anchor=base east,inner sep=0pt] at (Q.west) {$Q =\,$};
\node[anchor=base east,inner sep=0pt] at (T-1-1.base west) {$T =\,$};
\node at (Q.north east) {};
\node at (T.south west) {};

\foreach \i in {2,5,12,15}
{
	\node[anchor=base] at ($ (Q-1-\i.north) + (0,2ex) $) {\scriptsize \i};
}
\foreach \i in {2,5,10,13}
{
	\node[anchor=base] at ($ (T-\i-1.base west) - (2.5ex,0) $) {\scriptsize \i};
}
\end{tikzpicture}}
\caption{Gap-gap case.}
\end{subfigure}\begin{subfigure}{.5\textwidth}
\centering\scalebox{0.73}{\begin{tikzpicture}
	\tikzset{
		invisible/.style = {minimum size=0pt, inner sep=0pt}
	}

\matrix[
	matrix of nodes,
	minimum size=2mm,
	column sep={3.5mm,between origins},
	row sep={3.5mm,between origins},
	inner sep=1pt,
	nodes={anchor=base},
	anchor=south west,
	font=\ttfamily,
    ampersand replacement=\&
] (Q) at (0,0) {
A\&C\&A\&T\&C\&T\&G\&C\&C\&A\&A\&C\&A\&T\&A\&T\\
};

\matrix[
	matrix of nodes,
	minimum size=2mm,
	column sep={3.5mm,between origins},
	row sep={3.5mm,between origins},
	inner sep=1pt,
	nodes={anchor=base},
	anchor=north east,
	font=\ttfamily,
] (T) at (0,0) {
A\\C\\A\\T\\C\\T\\G\\C\\C\\A\\A\\C\\C\\C\\A\\A\\
};

\begin{scope}[on background layer]
\draw[black!15,step=3.5mm] (0,0) grid ($(Q-1-16.south east|-T-16-1.south east) + (1.1pt,-.01pt)$);
\end{scope}

\draw[-,very thick, colorA] ($(Q-1-2.north west) + (0,.05cm)$) -- node[invisible] (Qanchor0) {} ($(Q-1-10.north east) + (0,.05cm)$);
\draw[-,very thick, colorA] ($(T-2-1.north west) + (-.05cm,0)$) -- node[invisible] (Tanchor0) {} ($(T-10-1.south west) + (-.05cm,0)$);

\node[invisible] at ($(Q-1-2.west|-T-2-1.north) + (-1pt,1pt)$) (a0northwest) {};
\node[invisible] at ($(Q-1-10.west|-T-10-1.north) + (3.5mm,-3.5mm) + (-1pt,1pt)$) (a0southeast) {};
\draw[dashed] (a0northwest) -- ($(Q-1-16.east|-T-16-1.south) + (-1pt,1pt)$) node[pos=1,anchor=east] {$\diag(a) = 0$};
\draw[-,very thick,draw=colorA] (a0northwest) rectangle (a0southeast) node[invisible,pos=.5] (a0) {};

\draw[dashed,thick,-,colorA] (a0northwest) -- ($(Q-1-2.south west-|a0northwest)$);
\draw[dashed,thick,-,colorA] (a0southeast) -- ($(Q-1-10.south east-|a0southeast)$);
\draw[dashed,thick,-,colorA] (a0northwest) -- ($(T-2-1.north east|-a0northwest)$);
\draw[dashed,thick,-,colorA] (a0southeast) -- ($(T-10-1.south east|-a0southeast)$);

\draw[-,very thick, colorB] ($(Q-1-7.north west) + (0,.10cm)$) -- node[invisible] (Qanchor1) {} ($(Q-1-12.north east) + (0,.10cm)$);
\draw[-,very thick, colorB] ($(T-7-1.north west) + (-.10cm,0)$) -- node[invisible] (Tanchor1) {} ($(T-12-1.south west) + (-.10cm,0)$);

\node[invisible] at ($(Q-1-7.west|-T-7-1.north) + (-1pt,1pt)$) (a1northwest) {};
\node[invisible] at ($(Q-1-12.west|-T-12-1.north) + (3.5mm,-3.5mm) + (-1pt,1pt)$) (a1southeast) {};
\draw[-,very thick,draw=colorB] (a1northwest) rectangle (a1southeast) node[invisible,pos=.5] (a1) {};

\draw[dashed,thick,-,colorB] (a1northwest) -- ($(Q-1-7.south west-|a1northwest)$);
\draw[dashed,thick,-,colorB] (a1southeast) -- ($(Q-1-12.south east-|a1southeast)$);
\draw[dashed,thick,-,colorB] (a1northwest) -- ($(T-7-1.north east|-a1northwest)$);
\draw[dashed,thick,-,colorB] (a1southeast) -- ($(T-12-1.south east|-a1southeast)$);

\node[anchor=base,colorAdark] at (a0) {\backgroundcontour{$a$}};
\node[anchor=base,colorAdark] at (a0) {$a$};

\node[anchor=center,colorBdark] at (a1) {\backgroundcontour{$a'$}};
\node[anchor=center,colorBdark] at (a1) {$a'$};

\begin{scope}[on background layer]
\draw[draw=none,fill=black!30] ($(a0southeast)$) rectangle ($(a1northwest)$) node[pos=.5] (cost0) {};
\end{scope}
\node at (cost0) {\backgroundcontour{$\connect(a,a') = 0 + 0$}};
\node at (cost0) {$\connect(a,a') = 0 + 0$};

\begin{scope}[on background layer]
\draw[draw=none,fill=black!30] ($(Q-1-6.north west)!0.5!(Q-1-7.north east)$) rectangle ($(Q-1-10.south east)!0.5!(Q-1-11.south west)$);
\draw[draw=none,fill=black!30] ($(T-7-1.north west)$) rectangle ($(T-10-1.south east)$);
\end{scope}

\begin{scope}[on background layer]
\end{scope}

\node[anchor=base east,inner sep=0pt] at (Q.west) {$Q =\,$};
\node[anchor=base east,inner sep=0pt] at (T-1-1.base west) {$T =\,$};
\node at (Q.north east) {};
\node at (T.south west) {};

\foreach \i in {2,7,10,12}
{
	\node[anchor=base] at ($ (Q-1-\i.north) + (0,2ex) $) {\scriptsize \i};
}
\foreach \i in {2,7,10,12}
{
	\node[anchor=base] at ($ (T-\i-1.base west) - (2.5ex,0) $) {\scriptsize \i};
}
\end{tikzpicture}}
\caption{Perfect chain, overlap-overlap case.}
\end{subfigure}\\[.2cm]
\begin{subfigure}{.5\textwidth}
\centering\scalebox{0.73}{\begin{tikzpicture}
	\tikzset{
		invisible/.style = {minimum size=0pt, inner sep=0pt}
	}

\matrix[
	matrix of nodes,
	minimum size=2mm,
	column sep={3.5mm,between origins},
	row sep={3.5mm,between origins},
	inner sep=1pt,
	nodes={anchor=base},
	anchor=south west,
	font=\ttfamily,
    ampersand replacement=\&
] (Q) at (0,0) {
A\&C\&A\&T\&C\&A\&T\&A\&T\&A\&A\&G\&C\&A\&A\&T\\
};

\matrix[
	matrix of nodes,
	minimum size=2mm,
	column sep={3.5mm,between origins},
	row sep={3.5mm,between origins},
	inner sep=1pt,
	nodes={anchor=base},
	anchor=north east,
	font=\ttfamily,
] (T) at (0,0) {
A\\C\\A\\T\\C\\A\\A\\A\\T\\A\\T\\A\\A\\C\\A\\A\\
};

\begin{scope}[on background layer]
\draw[black!15,step=3.5mm] (0,0) grid ($(Q-1-16.south east|-T-16-1.south east) + (1.1pt,-.01pt)$);
\end{scope}

\draw[-,very thick, colorA] ($(Q-1-2.north west) + (0,.05cm)$) -- node[invisible] (Qanchor0) {} ($(Q-1-10.north east) + (0,.05cm)$);
\draw[-,very thick, colorA] ($(T-2-1.north west) + (-.05cm,0)$) -- node[invisible] (Tanchor0) {} ($(T-10-1.south west) + (-.05cm,0)$);

\node[invisible] at ($(Q-1-2.west|-T-2-1.north) + (-1pt,1pt)$) (a0northwest) {};
\node[invisible] at ($(Q-1-10.west|-T-10-1.north) + (3.5mm,-3.5mm) + (-1pt,1pt)$) (a0southeast) {};
\draw[dashed] (a0northwest) -- ($(Q-1-16.east|-T-16-1.south) + (-1pt,1pt)$) node[pos=1,anchor=south] {$\diag(a) = 0$};
\draw[-,very thick,draw=colorA] (a0northwest) rectangle (a0southeast) node[invisible,pos=.5] (a0) {};

\draw[dashed,thick,-,colorA] (a0northwest) -- ($(Q-1-2.south west-|a0northwest)$);
\draw[dashed,thick,-,colorA] (a0southeast) -- ($(Q-1-10.south east-|a0southeast)$);
\draw[dashed,thick,-,colorA] (a0northwest) -- ($(T-2-1.north east|-a0northwest)$);
\draw[dashed,thick,-,colorA] (a0southeast) -- ($(T-10-1.south east|-a0southeast)$);

\draw[-,very thick, colorB] ($(Q-1-6.north west) + (0,.10cm)$) -- node[invisible] (Qanchor1) {} ($(Q-1-11.north east) + (0,.10cm)$);
\draw[-,very thick, colorB] ($(T-8-1.north west) + (-.10cm,0)$) -- node[invisible] (Tanchor1) {} ($(T-13-1.south west) + (-.10cm,0)$);

\node[invisible] at ($(Q-1-6.west|-T-8-1.north) + (-1pt,1pt)$) (a1northwest) {};
\node[invisible] at ($(Q-1-11.west|-T-13-1.north) + (3.5mm,-3.5mm) + (-1pt,1pt)$) (a1southeast) {};
\draw[dashed] (a1northwest) -- ($(Q-1-14.east|-T-16-1.south) + (-1pt,1pt)$) node[pos=1,invisible] (diagonalbottom1) {};
\node[inner sep=0] at ($(diagonalbottom1) + (0,-1*3ex)$) {$\diag(a') = -2$};
\draw[-,very thick,draw=colorB] (a1northwest) rectangle (a1southeast) node[invisible,pos=.5] (a1) {};

\draw[dashed,thick,-,colorB] (a1northwest) -- ($(Q-1-6.south west-|a1northwest)$);
\draw[dashed,thick,-,colorB] (a1southeast) -- ($(Q-1-11.south east-|a1southeast)$);
\draw[dashed,thick,-,colorB] (a1northwest) -- ($(T-8-1.north east|-a1northwest)$);
\draw[dashed,thick,-,colorB] (a1southeast) -- ($(T-13-1.south east|-a1southeast)$);

\node[anchor=base,colorAdark] at (a0) {\backgroundcontour{$a$}};
\node[anchor=base,colorAdark] at (a0) {$a$};

\node[anchor=center,colorBdark] at (a1) {\backgroundcontour{$a'$}};
\node[anchor=center,colorBdark] at (a1) {$a'$};

\begin{scope}[on background layer]
\draw[draw=none,fill=black!30] ($(a0southeast)$) rectangle ($(a1northwest)$) node[pos=.5] (cost0) {};
\end{scope}
\node at (cost0) {\backgroundcontour{$\connect(a,a') = 0 + 2$}};
\node at (cost0) {$\connect(a,a') = 0 + 2$};

\begin{scope}[on background layer]
\draw[draw=none,fill=black!30] ($(Q-1-5.north west)!0.5!(Q-1-6.north east)$) rectangle ($(Q-1-10.south east)!0.5!(Q-1-11.south west)$);
\draw[draw=none,fill=black!30] ($(T-8-1.north west)$) rectangle ($(T-10-1.south east)$);
\end{scope}

\node[anchor=base east,inner sep=0pt] at (Q.west) {$Q =\,$};
\node[anchor=base east,inner sep=0pt] at (T-1-1.base west) {$T =\,$};
\node at (Q.north east) {};
\node at (T.south west) {};

\foreach \i in {2,6,10,11}
{
	\node[anchor=base] at ($ (Q-1-\i.north) + (0,2ex) $) {\scriptsize \i};
}
\foreach \i in {2,8,10,13}
{
	\node[anchor=base] at ($ (T-\i-1.base west) - (2.5ex,0) $) {\scriptsize \i};
}
\end{tikzpicture}}
\caption{Overlap-overlap case.}
\end{subfigure}%
\begin{subfigure}{.5\textwidth}
\centering\scalebox{0.73}{\begin{tikzpicture}
	\tikzset{
		invisible/.style = {minimum size=0pt, inner sep=0pt}
	}

\matrix[
	matrix of nodes,
	minimum size=2mm,
	column sep={3.5mm,between origins},
	row sep={3.5mm,between origins},
	inner sep=1pt,
	nodes={anchor=base},
	anchor=south west,
	font=\ttfamily,
    ampersand replacement=\&
] (Q) at (0,0) {
A\&C\&A\&T\&C\&T\&G\&C\&C\&A\&A\&C\&A\&T\&A\&T\\
};

\matrix[
	matrix of nodes,
	minimum size=2mm,
	column sep={3.5mm,between origins},
	row sep={3.5mm,between origins},
	inner sep=1pt,
	nodes={anchor=base},
	anchor=north east,
	font=\ttfamily,
] (T) at (0,0) {
A\\C\\A\\T\\C\\T\\G\\C\\C\\A\\A\\G\\C\\C\\A\\A\\
};

\begin{scope}[on background layer]
\draw[black!15,step=3.5mm] (0,0) grid ($(Q-1-16.south east|-T-16-1.south east) + (1.1pt,-.01pt)$);
\end{scope}

\draw[-,very thick, colorA] ($(Q-1-2.north west) + (0,.05cm)$) -- node[invisible] (Qanchor0) {} ($(Q-1-10.north east) + (0,.05cm)$);
\draw[-,very thick, colorA] ($(T-2-1.north west) + (-.05cm,0)$) -- node[invisible] (Tanchor0) {} ($(T-10-1.south west) + (-.05cm,0)$);

\node[invisible] at ($(Q-1-2.west|-T-2-1.north) + (-1pt,1pt)$) (a0northwest) {};
\node[invisible] at ($(Q-1-10.west|-T-10-1.north) + (3.5mm,-3.5mm) + (-1pt,1pt)$) (a0southeast) {};
\draw[dashed] (a0northwest) -- ($(Q-1-16.east|-T-16-1.south) + (-1pt,1pt)$) node[pos=1,anchor=south] {$\diag(a) = 0$};
\draw[-,very thick,draw=colorA] (a0northwest) rectangle (a0southeast) node[invisible,pos=.5] (a0) {};

\draw[dashed,thick,-,colorA] (a0northwest) -- ($(Q-1-2.south west-|a0northwest)$);
\draw[dashed,thick,-,colorA] (a0southeast) -- ($(Q-1-10.south east-|a0southeast)$);
\draw[dashed,thick,-,colorA] (a0northwest) -- ($(T-2-1.north east|-a0northwest)$);
\draw[dashed,thick,-,colorA] (a0southeast) -- ($(T-10-1.south east|-a0southeast)$);

\draw[-,very thick, colorB] ($(Q-1-7.north west) + (0,.10cm)$) -- node[invisible] (Qanchor1) {} ($(Q-1-11.north east) + (0,.10cm)$);
\draw[-,very thick, colorB] ($(T-12-1.north west) + (-.10cm,0)$) -- node[invisible] (Tanchor1) {} ($(T-16-1.south west) + (-.10cm,0)$);

\node[invisible] at ($(Q-1-7.west|-T-12-1.north) + (-1pt,1pt)$) (a1northwest) {};
\node[invisible] at ($(Q-1-11.west|-T-16-1.north) + (3.5mm,-3.5mm) + (-1pt,1pt)$) (a1southeast) {};
\draw[dashed] (a1northwest) -- ($(Q-1-11.east|-T-16-1.south) + (-1pt,1pt)$) node[pos=1,invisible] (diagonalbottom1) {};
\node[inner sep=0] at ($(diagonalbottom1) + (0,-1*3ex)$) {$\diag(a') = -5$};
\draw[-,very thick,draw=colorB] (a1northwest) rectangle (a1southeast) node[invisible,pos=.5] (a1) {};

\draw[dashed,thick,-,colorB] (a1northwest) -- ($(Q-1-7.south west-|a1northwest)$);
\draw[dashed,thick,-,colorB] (a1southeast) -- ($(Q-1-11.south east-|a1southeast)$);
\draw[dashed,thick,-,colorB] (a1northwest) -- ($(T-12-1.north east|-a1northwest)$);
\draw[dashed,thick,-,colorB] (a1southeast) -- ($(T-16-1.south east|-a1southeast)$);

\node[anchor=base,colorAdark] at (a0) {\backgroundcontour{$a$}};
\node[anchor=base,colorAdark] at (a0) {$a$};

\node[anchor=base,colorBdark] at (a1) {\backgroundcontour{$a'$}};
\node[anchor=base,colorBdark] at (a1) {$a'$};

\begin{scope}[on background layer]
\draw[draw=none,fill=black!30] ($(a0southeast)$) -- (a0southeast-|a1northwest) -- (a0southeast|-a1northwest) -- cycle;
\draw[draw=none,pattern={Lines[angle=45,distance=2pt,line width=1pt]},pattern color=black!30] ($(a1northwest)$) -- (a0southeast-|a1northwest) -- (a0southeast|-a1northwest) -- cycle;
\draw[draw=none] (a0southeast) -- (a1northwest) node[pos=.5] (cost0) {};
\end{scope}
\node at (cost0) {\backgroundcontour{$\mathrm{connect}(a_1,a_2) = 1 + 4$}};
\node at (cost0) {$\mathrm{connect}(a_1,a_2) = 1 + 4$};

\begin{scope}[on background layer]
\draw[draw=none,fill=black!30] ($(Q-1-6.north west)!0.5!(Q-1-7.north east)$) rectangle ($(Q-1-10.south east)!0.5!(Q-1-11.south west)$);
\draw[draw=none,pattern={Lines[angle=45,distance=2pt,line width=1pt]},pattern color=black!30] ($(T-11-1.north west)$) rectangle ($(T-11-1.south east)$);
\end{scope}

\node[anchor=base east,inner sep=0pt] at (Q.west) {$Q =\,$};
\node[anchor=base east,inner sep=0pt] at (T-1-1.base west) {$T =\,$};
\node at (Q.north east) {};
\node at (T.south west) {};

\foreach \i in {2,7,10,11}
{
	\node[anchor=base] at ($ (Q-1-\i.north) + (0,2ex) $) {\scriptsize \i};
}
\foreach \i in {2,10,12,16}
{
	\node[anchor=base] at ($ (T-\i-1.base west) - (2.5ex,0) $) {\scriptsize \i};
}
\end{tikzpicture}}
\caption{Overlap-gap case.}
\end{subfigure}
\caption{2D interpretation of exact match anchors and of function $\connect(a,a')$, where $a$ and $a'$ are represented as squares. The gap-gap case is represented with striped background, and the overlap in the other cases is represented with solid background.}
    \label{fig:diagonal}
\end{figure}
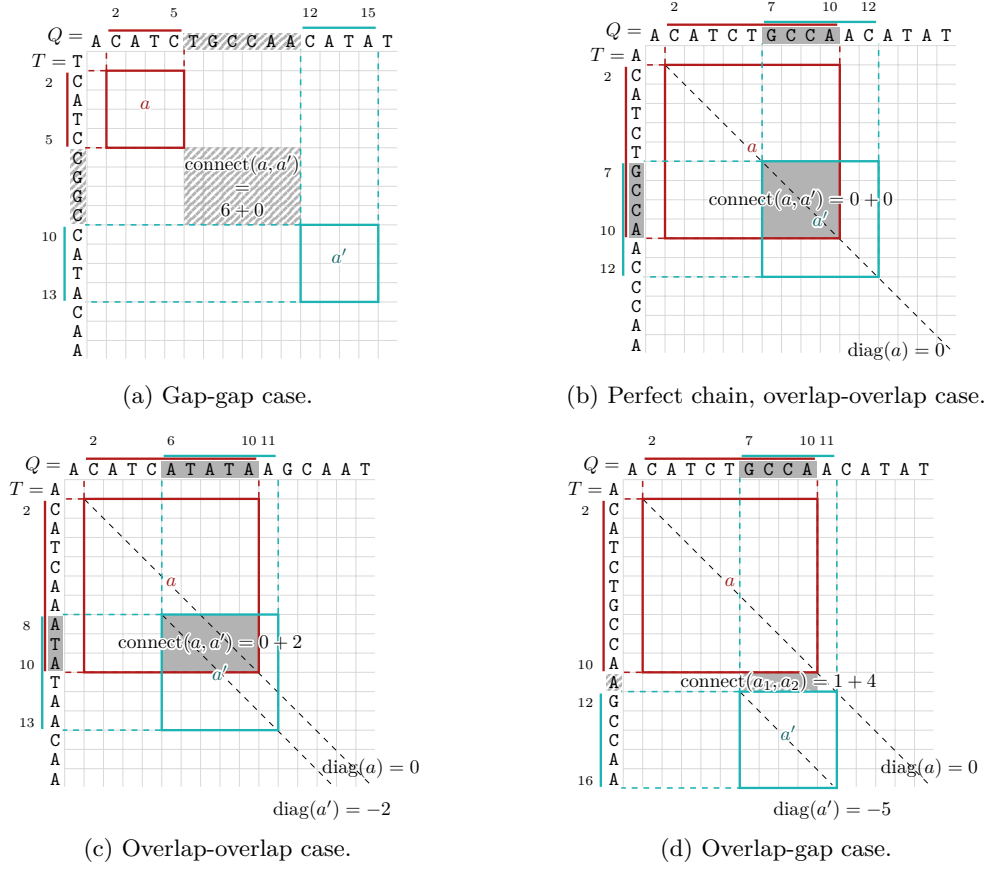

\begin{proof}
    From the exact match anchor definition, it follows that $\diag(a) = \qbeg - \tbeg = \qend - \tend$ (\Cref{obs:matchinvariant}).
    We can prove the thesis by considering all possible overlap cases for the $\connect$ function, knowing that $a \prec a'$ or $a \wprec a'$ by hypothesis:
\begin{itemize}[nosep]
    \item (gap-gap) $\connect(a,a') = \maxgap(a,a')$ directly from the definition of gap cost (\Cref{prob:CLC});
    \item (overlap-gap) if $a$ overlaps with $a'$ in $Q$ but not in $T$, in symbols $\qbeg' \le \qend \le \qend'$ and $\tend < \tbeg'$, we have that
    \begin{align*}
        \connect(a,a') &= \maxgap(a,a') + \diffoverlap(a,a') \\
        &= \tbeg' - \tend - 1 + \qend - \qbeg' + 1
        = (\qend - \tend) - (\qbeg' - \tbeg') \\
        &= \diag(a) - \diag(a')
        = \lvert \diag(a) - \diag(a') \rvert
    \end{align*}
    where the last equality follows due to the $\connect$ function being always greater than or equal to $0$ 
    (alternatively, note that $\diag(a) = \qend - \tend > \qbeg' - \tbeg' = \diag(a')$, since in the overlap-gap case $\qbeg' \le \qend$ and $\tend < \tbeg'$);
    \item (gap-overlap) symmetrically, if $a$ overlaps with $a'$ in $Q$ but not in $T$, we have that $\connect(a,a') = \diag(a') - \diag(a) > 0$;
    \item (overlap-overlap) if $a$ overlaps with $a'$ in both $Q$ and $T$, in symbols $\qbeg' \le \qend \le \qend'$ and $\tbeg' \le \tend \le \tend'$, we have that $\maxgap(a,a') = 0$ and
    \begin{align*}
        \connect(a,a') &= \diffoverlap(a,a')
        = \big\lvert (\qend - \qbeg' + 1) - (\tend - \tbeg' + 1) \big\rvert \\
        &= \big\lvert (\qend - \tend) - (\qbeg' - \tbeg') \big\rvert
        = \big\lvert \diag(a) - \diag(a') \,\big\rvert.
    \end{align*}
\end{itemize}
Thus, function $\connect$ can be described by only two cases: gap-gap and all other cases.
\end{proof}

Using \Cref{lem:diagonalcost}, we can rewrite $\mathrm{overlap}_Q(i)$ from \Cref{eq:recursion2} as
\begin{equation}
    \label{eq:overlap}
    \mathrm{overlap}_Q(i) =
    \min \Big\lbrace C[j] + \connect(a_j, a_i) \;\Big|\; j < i : \substack{a_j \prec a_i,\\ \qbeg^j \le \qbeg^i \le \qend^j,\\ \lvert \diag(a_j) - \diag(a_i) \rvert \le B} \Big\rbrace
\end{equation}
requiring us to consider only anchors $a_j \prec a_i$ such that: $(i)$ $j < i$ (recall that we assume anchors of $\mathcal{A}$ to be sorted by increasing order of $\qbeg$); $(ii)$ $a_j$ overlaps with $a_i$ in $Q$; and $(iii)$ the diagonal distance between $a_i$ and $a_j$ is at most $B$.
This can be obtained by preprocessing in linear time the diagonals of all anchors in $\mathcal{A}$ and putting the anchors in at most $n$ sorted buckets corresponding to their diagonals.

Our final optimization is obtained as follows: after $O(n \log n)$-time preprocessing, 
we can assume that for each diagonal $D$, there is at most one anchor $a_j$ such that $a_j$ overlaps $a_i$ in $Q$ and $\diag(a_j) = D$.

\begin{lemma}
    \label{assumption:diag}
    Given instance $\mathcal{A} = \lbrace a_1, \dots, a_n \rbrace$ of \Cref{prob:CLC} under \texttt{\upshape ChainX} precedence, we can obtain $\mathcal{A}'$ in $O(n \log n)$ time such that the anchored edit distance on $\mathcal{A}$ and $\mathcal{A}'$ coincide and $\mathcal{A}'$ is \emph{maximal under perfect chains}, defined as follows: for any anchors $a_i = ([\qbeg^i..\qend^i],[\tbeg^i..\tend^i]), a_j = ([\qbeg^j..\qend^j],[\tbeg^j..\tend^j]) \in \mathcal{A'}$ such that $a_i \chainxprec a_j$ and 
    $\connect(a_i,a_j) = 0$, we have that $\qend^i < \qbeg^j$ and $\tend^i < \tbeg^i$.
\end{lemma}
\begin{proof}
    First, note that from \Cref{lem:diagonalcost} and \Cref{obs:matchinvariant} it follows that $\connect(a_i,a_j) = 0$ ($a_i \chainxprec a_j$), that is, $a_i$ and $a_j$ form a perfect chain, if and only if $\diag(a_i) = \diag(a_j)$ and $\qend  = \qbeg' - 1$ (gap-gap case) or $\qend > \qbeg' - 1$ (other cases, see \Cref{fig:diagonal} (b)).
    We construct $\mathcal{A'}$ by merging all anchors of $\mathcal{A}$ forming perfect chains as in \Cref{cor:maximality} in $O(n \log n)$ time, by first sorting $\mathcal{A}$ by start position $\qbeg$ in $Q$, then by stable sorting $\mathcal{A}$ by diagonal $\diag(a)$, and merging adjacent anchors with $\connect$ value equal to $0$.
    The procedure terminates and merging order does not matter, since the binary relation $\mathrm{R} \subseteq \mathcal{A} \times \mathcal{A}$ containing pairs $(a,a')$ such that $a,a'$ are part of a perfect chain of $\mathcal{A}$ is an equivalence.
    The correctness follows from \Cref{cor:maximality}.
\end{proof}

\begin{algorithm}[htp]
\KwIn{Anchors $\mathcal{A} = \lbrace a_1, \dots, a_n \rbrace$ between $Q \in \Sigma^+$ and $T \in \Sigma^+$ that are maximal under perfect chains (\Cref{assumption:diag}) and parameters $B_\mathrm{start}$ (initial guess), $\alpha > 1$ (ramp-up factor).}
\KwOut{Maximum cost $\gscost(A)$ of a chain $A = \overline{a}_1, \dots, \overline{a}_c$ under $\chainxprec$.}
Sort pairs $(\qbeg, i)$ and $(\qend, i)$, for $a_i = ([\qbeg..\qend], [\tbeg..\tend]) \in \mathcal{A}$, into array $\mathtt{A}[1..2n]$ by first component\;
Sort pairs $(\diag(a), a)$, for $a = ([\qbeg..\qend], [\tbeg..\tend]) \in \mathcal{A}$, into non-empty buckets $\mathcal{D}_1$, \dots, $\mathcal{D}_d$ by first component, with the buckets ordered by increasing $\diag(a)$, and let $\bucket(a) = f$ if anchor $a$ belongs to the $f$-th bucket, and $\diag(f) = \diag(a)$ for any anchor $a$ in bucket $\mathcal{D}_f$\;
$B \gets B_\mathrm{start}$\;
Initialize $\mathtt{C}[1..n+1]$ to values $+\infty$\tcc*{recursive values from \Cref{eq:recursion2}}
Initialize $\mathtt{D}[1..d]$ to values $\perp$\tcc*{$\mathtt{D}[f] = i$ if anchor $a_i \in \mathcal{D}_f$ is active}
\Repeat{$\mathtt{C}[n+1] \le B_\mathrm{last}$}{
    $\mathtt{C}[n+1] \gets +\infty$\;
    \For(\tcp*[f]{iterate over all startpoint $\qbeg$ and endpoints $\qend$}){$k \gets 1 \;\KwTo\; 2n$}{%
        $(q,i) \gets \mathtt{A}[k]$\;
        $([\qbeg..\qend], [\tbeg..\tend]) \gets a_i$\;
        \If(\tcp*[f]{if startpoint, update active anchor and compute $C[i]$}){$q = \qbeg$}{%
            $\mathtt{D}[\bucket(a_i)] \gets i$\;
            \uIf(\tcp*[f]{initial anchor $a_{0} = \startanchor$}){$\connect(a_0, a_i) \le B$}{%
                $\mathtt{C}[i] \gets \connect(a_0, a_i)$\;
            }\Else{
                $\mathtt{C}[i] \gets +\infty$\;
            }
            \For(\tcp*[f]{iterate over all endpoints $\qend^j$ at distance $\le B$}){$k' \gets i - 1 \;\KwDownTo\; 1$}{%
                $(q',j) \gets A[k']$\;
                $([\qbeg'..\qend'], [\tbeg'..\tend']) \gets a_j$\;
                \If(\tcp*[f]{check distance}){$\qbeg - \qend' > B$}{%
                    \KwBreak\;
                }
                \If{$q' = \qend'$ \text{and} $a_j \prec a_i$ \text{and} $\connect(a_j,a_i) \le B$}{%
                    $\mathtt{C}[i] \gets \min\big( \mathtt{C}[i], \mathtt{C}[j] + \connect(a_j,a_i) \big)$\tcc*{compute $\mathrm{gap}_Q(i)$}
                }
            }
            \For{$f \in [1..d] : \lvert \diag(f) - \diag(a_i) \rvert \le B$}{
                \If{$\mathtt{D}[f] \neq \perp$}{%
                    $\mathtt{C}[i] \gets \min\big( \mathtt{C}[i], \mathtt{C}[\mathtt{D}[f]] + \lvert \diag(f) - \diag(a_i) \rvert) \big)$\tcc*{compute $\mathrm{overlap}_Q(i)$}
                }
            }
            \If(\tcp*[f]{final anchor $a_{n+1} = \finalanchor$}){$\connect(a_i, a_{n+1}) \le B$}{
                $\mathtt{C}[n+1] \gets \min \big( \mathtt{C}[n+1], \mathtt{C}[i] + \connect(a_i,a_{n+1}) \big)$\;
            }
        }
        \If(\tcp*[f]{if endpoint, remove as active anchor}){$q = \qend$}{%
            $\bucket(a_i) \gets \perp$\;
        }
    }
    $B_\mathrm{last} \gets B$\;
    $B \gets B \cdot \alpha$\;
}

\Return{$\mathtt{C}[n+1]$}\;
\caption{Practical colinear chaining on sequences revisited.
\label{alg:2}}
\end{algorithm}

\begin{theorem}
    \label{thm:practicalCLCrevisited}
    Given anchors $\mathcal{A} = a_1, \dots, a_n$ between $Q \in \Sigma^+$ and $T \in \Sigma^+$ \Cref{alg:2} solves colinear chaining with overlap and gap costs (\Cref{prob:CLC}) under $\chainxprec$ on $\mathcal{A}$ in $O(\mathrm{OPT} \cdot n + n \log n)$ average-case time, assuming that $n \le \lvert Q \rvert$, the anchors are uniformly distributed in $Q$, and the anchor set is maximal under perfect chains.
\end{theorem}
\begin{proof}
    The correctness follows from \Cref{eq:recursion2,eq:overlap}: the inner loop (lines 8--30) computes $C[i]$ by considering only and all anchors $a_j$ with $\connect(a_j,a_i) \le B$; in particular, lines 17--23 compute $\mathrm{gap}_Q(i)$ by considering all anchor startpoints and endpoints at distance at most $B$ in $Q$, and lines 24--26 compute $\mathrm{overlap}_Q(i)$ by considering all diagonals at distance at most $B$ from $\diag(a_i)$ in the sorted sequence of diagonals computed in the form of buckets (line 2).
    The algorithm maintains at most one active anchor per diagonal (lines 12 and 30) due to the maximality under perfect chains (\Cref{assumption:diag}): for any given startpoint $\qbeg$ and diagonal, there can be only one active anchor.

    The $O(n \log n)$ term comes from sorting and bucketing the anchors in lines 1--2.
    Each iteration of the main loop (lines 6--33) takes $O(n B)$ time: the handling of diagonals takes $O(n)$ time in total; at most $2B + 2$ diagonals are considered for each $C[i]$ in the computation of $\mathrm{gap}_Q(i)$; and the average-case analysis of \cite[Lemma 6]{JainGT22} also holds for lines 17--23.
    Indeed, let $Y_{j,q}$ be the indicator random variable equal to $1$ if anchor $a_j$ ends at position $q$ of $Q$, in symbols $\qend^j = q$, and $0$ otherwise.
    Then, $\mathbb{E}[Y_{j,q}] = 1 / \lvert Q \rvert$ for all $j \in [1..n]$ and $q \in [1..\lvert Q \rvert]$ due to the uniform distribution assumption.
    Similarly, if $X_{j,q}$ is defined as $Y_{j,q}$ but with condition $\qbeg^j = q$ (anchor $a_j$ starts at position $q$), $\mathbb{E}[X_{j,q}] = 1 / \lvert Q \rvert$.\footnote{The algorithm can be further engineered to avoid these comparisons and the use of $X_{j,q}$. However the average-case time complexity analysis stays the same.}
    If $Z^i$ is the number of anchor startpoints $\qbeg^j$ such that $\qbeg^i - B \le \qbeg^j \le \qbeg^i$ plus the number of endpoints $\qend^j$ with $\qbeg^i - B \le \qend^j \le \qbeg^i$, then
    $
        Z^i = \sum_{j=1}^{n} \sum_{q = \qbeg^i - B}^{\qbeg^i} \big( X_{j,q} + Y_{j,q} \big)
    $.
    Finally, if $Z$ is the total number of anchors processed in lines 17--23, then
    \begin{align*}
        \mathbb{E}(Z) &= \mathbb{E}\bigg(
            \sum_{i = 1}^{n} Z^i 
        \bigg)\\
        &= \sum_{i=1}^{n} \sum_{j=1}^{n}\sum_{q = \qbeg^i - B}^{\qbeg^i} \big( \mathbb{E}(X_{j,q}) + \mathbb{E}(Y_{j,q}) \big) &\text{linearity of expectation}\\
        &= \frac{n^2 \cdot 2(B+1)}{\lvert Q \rvert} &\mathbb{E}\big(X_{j,q}), \mathbb{E}\big(Y_{j,q}) = \frac{1}{\lvert Q \rvert} \\
        &\le 2n(B+1) &n \le \lvert Q \rvert
    \end{align*}
    and thus the loop 17-23 takes $O(nB)$ average time.
    The time complexity of the main computation (lines 6-34) is then
$
    O\Big(B_\mathrm{start} \cdot n \cdot \Big( 1 + \alpha + \alpha^2 + \dots + \alpha^{\lceil \log_\alpha \mathrm{OPT} \rceil} \Big)\!\Big) = O(n \cdot \mathrm{OPT}).
$
\end{proof}

Given anchor set $\mathcal{A}$, it takes $O(\mathrm{OPT}\cdot n + n \log n)$ total time to preprocess it as per \Cref{assumption:diag} and solve the chaining problem with \Cref{thm:practicalCLCrevisited}.
We leave the proof that chaining under strict precedence $\prec$ is equivalent to chaining under \texttt{\upshape ChainX} precedence $\chainxprec$, implying that the transformation of \Cref{assumption:diag} maintains the optimal chain cost under $\chainxprec$, to \Cref{sec:appendix} of the Appendix.

\section{Experiments}\label{sec:experiments}
We implemented \Cref{alg:2} and integrated it in a fork of the C++ tool \texttt{\upshape ChainX}, available at \url{https://github.com/algbio/ChainX}.
This provably correct version of \texttt{\upshape ChainX} can be invoked with flag \texttt{{-}{-}optimal}.
We replicated the original experiment of \texttt{\upshape ChainX} \cite[Table 2]{JainGT22} in \Cref{table:2} to compare \texttt{\upshape ChainX} to the \texttt{{-}{-}optimal} version, that we denote as \texttt{\upshape ChainX-opt}, on the University of Helsinki cluster \texttt{ukko}, limiting the task to 64 GB of memory, 4 cores, using the Lustre Vakka cluster filesystem (comparable to SSD performance).
Since the original parameters $B_\mathrm{start} = 100$ (initial guess) and $\alpha = 4$ (ramp-up factor) were hard-coded constants, we tested the variant of $\mathtt{ChainX\textnormal{-}opt}$ that we denote as \texttt{\upshape ChainX-opt*} using variable $B_\mathrm{start}$ equal to the maximum between $100$ and the inverse coverage of query $Q$ (number of bases that are not covered by any anchor) multiplied by $1.5$.
A similar optimization was previously introduced and used in \texttt{chainx-block-graph}~\cite{rizzo2024exploiting,SRFAligner}.
All costs of the output chains in the original experiments have the same exact value as reported by \texttt{\upshape ChainX-opt} (and \texttt{\upshape ChainX-opt*}), verifying the results of the work by Jain et al.
The implementation of \Cref{alg:2} is, naturally, slightly more computationally expensive, however the optimization introduced for \texttt{\upshape ChainX-opt*} appears to mitigate the additional computations.
We successfully replicated \cite[Table 1]{JainGT22} with analogous results in \Cref{table:1}.

\begin{table}[htp]
    \caption{Replication of~\cite[Table 2]{JainGT22}, using MUM seeds of length at least $20$. Label avg.\ iters refers to the average number of iterations of the main loop of each algorithm. All modes output a chain with the same cost, verifying the original results.}\label{table:2}
    \centering
    \setlength{\tabcolsep}{4pt}\begin{tabular}{c|r|r|r|r|r|r|r|r|r}
        \toprule
	\multirow{2}*{Similarity} & \multicolumn{3}{c}{\texttt{\upshape ChainX}} & \multicolumn{3}{c}{\texttt{\upshape ChainX-opt}} & \multicolumn{3}{c}{\texttt{\upshape ChainX-opt*}} \\
        \cmidrule(lr){2-4}
        \cmidrule(lr){5-7}
        \cmidrule(lr){8-10}
        & \scriptsize time (s) & \scriptsize space (MB) & \scriptsize avg.\ iters & \scriptsize time (s) & \scriptsize space (MB) & \scriptsize avg.\ iters & \scriptsize time (s) & \scriptsize space (MB) & \scriptsize avg.\ iters \normalsize \\
        \midrule
        \multicolumn{10}{l}{Semiglobal sequence comparison, sequence sizes $10^4$ (100 queries) and $5 \cdot 10^6$ (reference)} \\
        \midrule
        90--100\% & 0.89 & 57.816 & 3.87 & 0.91 & 57.592 & 3.87 & 0.87 & 57.788 & 1.00 \\
        80--90\%  & 0.98 & 57.672 & 5.00 & 0.98 & 57.824 & 5.00 & 1.02 & 57.604 & 1.00 \\
        75--80\%  & 1.01 & 57.464 & 5.00 & 0.98 & 57.808 & 5.00 & 1.03 & 57.820 & 1.00 \\
        \midrule
        \multicolumn{10}{l}{Global sequence comparison, sequence sizes $10^6$ (100 queries and reference)} \\
        \midrule
        90--100\% & 13.53 & 120.976 & 8.00 & 12.84 & 121.104 & 8.00 & 12.28 & 120.820 & 1.00 \\
        80--90\%  & 27.37 & 120.932 & 8.00 & 37.20 & 120.812 & 8.00 & 23.38 & 121.092 & 1.00 \\
        75--80\% & 13.53 & 120.976 & 8.00 & 12.84 & 121.104 & 8.00 & 12.28 & 120.820 & 1.00 \\
        \bottomrule
    \end{tabular}
\end{table}
\begin{table}[htp]
    \caption{Replication of~\cite[Table 1]{JainGT22}, using MUM seeds of length at least $20$. Label \emph{iters} refers to the number of iterations of the main loop of each algorithm. All modes output a chain with the same cost, verifying the original results.}\label{table:1}
    \centering
    \setlength{\tabcolsep}{4pt}\begin{tabular}{c|r|r|r|r|r|r|r|r|r|r}
        \toprule
        \multirow{2}{*}{Similarity} & \multirow{2}*{MUM count} & \multicolumn{3}{c}{\texttt{\upshape ChainX}} & \multicolumn{3}{c}{\texttt{\upshape ChainX-opt}} & \multicolumn{3}{c}{\texttt{\upshape ChainX-opt*}} \\
        \cmidrule(lr){3-5}
        \cmidrule(lr){6-8}
        \cmidrule(lr){9-11}
        && \scriptsize time (s) & \scriptsize space (MB) & \scriptsize iters & \scriptsize time (s) & \scriptsize space (MB) & \scriptsize iters & \scriptsize time (s) & \scriptsize space (MB) & \scriptsize iters \\
        \midrule
        \multicolumn{11}{l}{Semiglobal sequence comparison, sequence sizes $10^4$ (query) and $5 \cdot 10^6$ (reference)} \\
        \midrule
        97\% & 160 & 0.003 & 56.776 & 3 & 0.003 & 56.820 & 3 & 0.003 & 56.604 & 1 \\
        94\% & 176 & 0.003 & 56.336 & 4 & 0.003 & 56.616 & 4 & 0.003 & 56.592 & 1 \\
        90\% & 135 & 0.004 & 56.464 & 4 & 0.004 & 56.636 & 4 & 0.004 & 56.948 & 1 \\
        80\% &  28 & 0.004 & 56.628 & 5 & 0.004 & 56.820 & 5 & 0.004 & 56.636 & 1 \\
        70\% &   3 & 0.004 & 56.752 & 5 & 0.004 & 56.836 & 5 & 0.004 & 56.776 & 1 \\
        \midrule
        \multicolumn{11}{l}{Global sequence comparison, sequence sizes $10^6$ (query and reference)} \\
        \midrule
        99\% &    7012 & 0.038 & 23.264 & 5 & 0.047 & 23.032 & 5 & 0.039 & 23.284 & 1 \\
        97\% & 15\,862 & 0.381 & 23.588 & 7 & 0.624 & 23.652 & 7 & 0.248 & 23.580 & 1 \\
        94\% & 18\,389 & 0.514 & 24.056 & 7 & 0.829 & 24.444 & 7 & 0.652 & 24.140 & 1 \\
        90\% & 14\,472 & 0.671 & 23.580 & 8 & 1.060 & 23.652 & 8 & 0.562 & 23.528 & 1 \\
        80\% &  2\,964 & 0.138 & 23.248 & 8 & 0.158 & 23.024 & 8 & 0.132 & 23.104 & 1 \\
        70\% &     195 & 0.115 & 23.024 & 8 & 0.118 & 23.264 & 8 & 0.115 & 23.244 & 1 \\
        \bottomrule
    \end{tabular}
\end{table}

We additionally tested the semiglobal mode of \texttt{\upshape ChainX} and \texttt{\upshape ChainX-opt*} on the T2T-CHM13 reference~\cite{Nurk22T2T,rhie2023complete} and a sample of 100k PacBio HiFi (run m64004\footnote{Available at \url{https://s3-us-west-2.amazonaws.com/human-pangenomics/index.html?prefix=T2T/scratch/HG002/sequencing/hifi/}.}) reads used in the assembly of the HG002 reference~\cite{rhie2023complete}.
On this dataset, \texttt{\upshape ChainX} and \texttt{\upshape ChainX-opt*} took
1\,610
and
1\,631
seconds, respectively, 49 GB of memory, and had an average number of iterations of the main loop of
4.39
and
2.61, respectively. 
\texttt{\upshape ChainX-opt*} improved the optimal chaining cost of
2\,297
long reads out of 100\,000 (2.30 \%).
For these improved reads, the distribution of the absolute value improvement---the \texttt{\upshape ChainX-opt*} cost minus the \texttt{\upshape ChainX} cost---presents a minimum of $11$, first quartile of $835$, median of $1888$, third quartile of $4080$, and maximum of $6394$.
The relative improvement---\texttt{\upshape ChainX-opt*} cost divided by the \texttt{\upshape ChainX} cost---presents statistics of $1.04$ (minimum), $18.39$ (first quartile), $186.00$ (median), $663.86$ (third quartile), and $6293.00$ (maximum). These results show that cases where the \texttt{\upshape ChainX} solution fails occur in realistic data.

\section*{Acknowledgements}

This project has received funding from the European Union’s Horizon Europe research and innovation programme under grant agreement No.\ 101060011 (TeamPerMed) and from the Helsinki Institute for Information Technology (HIIT).

\bibliographystyle{plainurl}
\bibliography{biblio}

\clearpage
\appendix

\section{Full equivalence to anchored edit distance}
\label{sec:appendix}
After having introduced the diagonal distance and the corresponding lemma linking it to the $\connect$ function (\Cref{lem:diagonalcost}), we can prove that, for the goal of finding an optimal chain, strong anchor precedence $\chainxprec$ is equivalent to strict anchor precedence $\prec$ (see \Cref{def:precedence}) and thus the chaining formulation of \Cref{eq:recursion,eq:recursion2} connects to the anchored edit distance.
In turn, this would imply that the transformation from \Cref{assumption:diag} does not change the cost of an optimal chain under $\chainxprec$.

\begin{theorem}\label{thm:chainxprecequivalence}
    For a fixed set of anchors $\mathcal{A}$, the optimal colinear chaining cost under \texttt{\upshape ChainX} precedence ($\chainxprec$) equals the anchored edit distance and optimal colinear chaining cost under strict or weak precedence.
\end{theorem}
\begin{proof}
    Recall that any chain under $\chainxprec$ is also a chain under $\prec$.
    The proof proceeds by applying \Cref{thm:edconnection} after proving that an optimal chain under strict precedence $\prec$ that \emph{is not} a chain under strong precedence $\chainxprec$ can be transformed into a chain under $\chainxprec$ of equal optimal cost (or even smaller cost, reaching a contradiction).
	Indeed, without loss of generality, this case occurs when an optimal chain $A = a_1 \cdots a_c$ is such that two adjacent anchors have the same starting position in $Q$, in symbols $\qbeg^i = \qbeg^{i+1}$ for $i \in [1..c-1]$, $a_{i} = ([\qbeg^i..\qend^i], [\tbeg^i..\tend^i])$, and $a_{i+1} = ([\qbeg^{i+1}..\qend^{i+1}], [\tbeg^{i+1}..\tend^{i+1}])$: since the colinear chaining problem definition (\Cref{prob:CLC}) is symmetrical with respect to both ($i$) the two input strings $Q$, $T$, and ($ii$) to the direction considered (the left-to-right precedence and cost definitions can be mirrored into equivalent right-to-left precedence and cost), it follows that an optimal chain under $\prec$ can be easily transformed into an optimal chain under $\chainxprec$.

    Then, let $A = a_1 \cdots a_c$ be a chain under $\prec$ but not under $\chainxprec$ with $\qbeg^i = \qbeg^{i+1}$ for some $a_{i} = ([\qbeg^i..\qend^i], [\tbeg^i..\tend^i])$ and $a_{i+1} = ([\qbeg^{i+1}..\qend^{i+1}], [\tbeg^{i+1}..\tend^{i+1}])$, with $i \in [1..c-1]$.
    Since $a_i \prec a_{i+1}$, we have that $(a)$ $\qbeg^i = \qbeg^{i+1}$, $(b)$ $\qend^i \le \qend^{i+1}$, $(c)$ $\tbeg^i \le \tbeg^{i+1}$, and $(d)$ $\tend^i \le \tend^{i+1}$.
    From $(a)$ we derive that $q_i$ and $q_{i+1}$ overlap in $Q$, and thus
\begin{align*}
    \connect(a_i, a_{i+1}) &=
    \lvert \diag(a_i) - \diag(a_{i+1}) \rvert &\text{\Cref{lem:diagonalcost}}\\
    &= \lvert (\qbeg^i - \tbeg^i) - (\qbeg^{i+1} - \tbeg^{i+1}) \rvert &\text{$\diag$ def.}\\
    &= \lvert \tbeg^{i+1} - \tbeg^i \rvert &(a) \\
    &= \tbeg^{i+1} - \tbeg^i \ge 0 &(c)
\end{align*}
    and for simplicity we call the non-negative integer quantity $\tbeg^{i+1} - \tbeg^i = \alpha$.
Then, we proceed to prove that $\connect(a_{i-1}, a_i) + \connect(a_{i}, a_{i+1}) = \connect(a_{i-1}, a_i) + \alpha  \ge \connect(a_{i-1}, a_{i+1})$ (i.e.\ skipping anchor $a_i$ does not increase the cost of the chain), with $a_{i-1} = ([\qbeg^{i-1}..\qend^{i}],[\tbeg^{i-1}..\tend^{i}])$. We analyze the two cases of \Cref{lem:diagonalcost} for anchor pair $a_{i-1}$, $a_i$:
\begin{itemize}
    \item (gap-gap case) if $a_{i-1}$ and $a_i$ do not overlap in neither $Q$ nor $T$, in symbols $\qend^{i-1} < \qbeg^{i}$ and $\tend^{i-1} < \tbeg^{i}$, then $a_{i-1}$, $a_{i+1}$ are also in the gap-gap case and
\begin{align*}
    \connect(a_{i-1}, a_{i}) + \alpha &= \maxgap(a_{i-1}, a_{i}) + \alpha &\text{\Cref{lem:diagonalcost}} \\
    &= \max\big( \qbeg^{i} - \qend^{i-1} - 1,\; \tbeg^{i} - \tend^{i-1} - 1 \big) + \alpha &\text{$\maxgap(\cdot)$ def.} \\
    &\ge \max\big( \qbeg^{i} - \qend^{i-1} - 1,\; \tbeg^{i} - \tend^{i-1} - 1  + \alpha \big) \\
    &= \max\big( \qbeg^{i} - \qend^{i-1} - 1,\; \tbeg^{i+1} - \tend^{i-1} - 1 \big) &\text{$\alpha$ def.} \\
    &= \max\big( \qbeg^{i+1} - \qend^{i-1} - 1,\; \tbeg^{i+1} - \tend^{i-1} - 1 \big) &(a) \\
    &= \maxgap(a_{i-1}, a_{i+1}) \\
    &= \connect(a_{i-1}, a_{i+1}) &\text{\Cref{lem:diagonalcost}};
\end{align*}

    \item (all other cases) if $a_{i-1}$ and $a_i$ overlap in $Q$ or in $T$, in symbols $\qend^{i-1} \ge \qbeg^i$ or $\tend^{i-1} \ge \tbeg^i$, then we need to consider the mutually exclusive cases of \Cref{lem:diagonalcost} for pair $a_{i-1}$, $a_{i+1}$:
    \begin{itemize}
        \item (gap-gap case) if $a_{i-1}$ and $a_{i+1}$ do not overlap neither in $Q$ nor in $T$, then $a_{i-1}$ and $a_{i}$ must overlap in $T$ but not in $Q$ (otherwise we would reach a contradiction), in symbols $\qend^{i-1} < \qbeg^{i} = \qbeg^{i+1}$ (due to $(a)$), $\tend^{i-1} \ge \tbeg^{i}$, and
\begin{align*}
    \connect(a_{i-1},a_{i}) + \alpha &=
    \maxgap(a_{i-1},a_{i}) + \diffoverlap(a_{i-1}, a_{i}) + \alpha &\text{$\connect$ def.} \\
    &= \qbeg^{i} - \qend^{i-1} - 1 + \diffoverlap(a_{i-1}, a_{i}) + \alpha &\text{gap-overlap case}\\
    &\ge \max(\qbeg^{i} - \qend^{i-1} - 1, \alpha) \\
    &= \max(\qbeg^{i+1} - \qend^{i-1} - 1, \alpha) &\text{$(a)$}\\
    &= \max(\qbeg^{i+1} - \qend^{i-1} - 1,\; \tbeg^{i+1} - \tbeg^{i}) &\text{$\alpha$ def.} \\
    &\ge \max(\qbeg^{i+1} - \qend^{i-1} - 1,\; \tbeg^{i+1} - \tend^{i-1}) &\tend^{i-1} \ge \tbeg^{i} \\
    &\ge \max(\qbeg^{i+1} - \qend^{i-1} - 1,\; \tbeg^{i+1} - \tend^{i-1} - 1) \\
    &= \maxgap(a_{i-1}, a_{i+1}) = \connect(a_{i-1},a_{i+1}) &\text{\Cref{lem:diagonalcost}};
\end{align*}
        \item (all other cases) if $a_{i-1}$ and $a_{i+1}$ overlap in $Q$ or in $T$ then
\begin{align*}
    \connect(a_{i-1}, a_i) + \alpha &=
    \lvert \diag(a_{i-1}) - \diag(a_{i}) \rvert + \lvert \diag(a_{i}) - \diag(a_{i+1}) \rvert &\text{\Cref{lem:diagonalcost}}\\
    &\ge \lvert \diag(a_{i-1}) - \diag(a_{i+1}) \rvert &\text{triangle ineq.}\\
    &= \connect(a_{i-1}, a_{i+1}) &\text{\Cref{lem:diagonalcost}},
\end{align*}
    completing the proof that skipping anchor $a_{i}$ does not increase the chaining cost $\gscost(A)$.\qedhere
\end{itemize}
\end{itemize}
\end{proof}

\Cref{thm:chainxprecequivalence} shows that the colinear chaining problem under \texttt{\upshape ChainX} precedence connects to anchored edit distance like the other formulations, thus the transformation of \Cref{assumption:diag} of the input anchors, making the anchor set maximal under perfect chains, does not change the cost of the optimal chain even under $\chainxprec$, obtaining the following corollary.
\begin{corollary}
\Cref{assumption:diag} and \Cref{alg:2} solve the anchored edit distance problem in $O(\mathrm{OPT} \cdot n + n \log n )$ average-case time, assuming that $n \le \lvert Q \rvert$ and the anchors are uniformly distributed in $Q$.
\end{corollary}
\noindent An important future extension of \Cref{thm:edconnection,thm:chainxprecequivalence} is to the sequence-to-graph setting: even though some chaining formulations have been adapted~\cite{LiFC20,ChandraJ23jcb,MaCSMT23bioinf,RizzoCM23spire}, no algorithm maintains the connection to edit distance.
The dynamic-programming strategy of Jain et al.\ as revisited in this work motivates the future extension of chaining approaches based on overlaps and gaps.
\end{document}